\documentclass[10pt]{article}

\usepackage{amssymb} 
\usepackage{latexsym} 
\usepackage{amsmath}
\usepackage{mathtools} 
\usepackage{tabls} 
\usepackage{graphicx}
\usepackage{todonotes}
\usepackage{overpic}
\usepackage{subcaption}

\usepackage{geometry}

\usepackage[utf8]{inputenc}
\usepackage{color}
\usepackage{epsfig}
\usepackage{graphicx}
\usepackage[all]{xy}
\input xy
\usepackage[mathcal]{eucal}
\usepackage{enumerate}
\usepackage{enumitem}
\usepackage{hyperref}
\usepackage{accents}

\newcommand{\subscript}[1]{{\rm S$_ #1$}}

\newtheorem{theorem}{Theorem}
\newtheorem{definition}{Definition}
\newtheorem{remark}{Remark}

\newtheorem{lemma}{Lemma}
\newtheorem{corollary}[theorem]{Corollary}

\newenvironment{proof}{\noindent{\bf Proof.}}%
{\hspace*{\fill}$\Box$\par\vspace{4mm}}

\newenvironment{claimproof}[2]{\par\noindent\underline{Proof of #1:}\space#2}{\hfill $\diamondsuit$\mybreak\noindent}

\def\cadre{$$\vcenter\bgroup\advance\hsize by -2em\noindent
	\refstepcounter{equation}(\theequation)~\ignorespaces}
\makeatletter
\def\endcadre{\egroup\eqno$$\global\@ignoretrue}

\newcommand{\mybreak} {\par\vspace{2mm}\noindent}

\newcommand{\Upper}{{\rm Upper}}

\newcommand{\cross}{{\rm cross}}
\newcommand{\intra}{{\rm intra}}
\newcommand{\cro}{\mathbf{\mathsf{h}}}
\newcommand{\ccro}{\cro_Q^\cross}

\begin{document}
	\pagestyle{plain}
	
	\title{A New Characterization of Path Graphs}
	\date{}
	\author{Nicola Apollonio\footnote{Istituto per le Applicazioni del
			Calcolo, M. Picone, Via dei Taurini 19, 00185 Roma, Italy.
			\texttt{nicola.apollonio@cnr.it}} \and
		{Lorenzo Balzotti\footnote{Dipartimento di Scienze di Base e Applicate per l’Ingegneria, Sapienza Universit\`a di Roma, Via Antonio Scarpa, 16,
				00161 Roma, Italy. \texttt{lorenzo.balzotti@sbai.uniroma1.it}.}}
	}
	
	\maketitle
	
	\begin{abstract} Path graphs are intersection graphs of paths in a tree.~In this paper we give a ``good characterization'' of path graphs, namely, we prove that path graph membership is in $NP\cap CoNP$ without resorting to existing polynomial time algorithms. The characterization is given in terms of the collection of the \emph{attachedness graphs} of a graph, a novel device to deal with the connected components of a graph after the removal of clique separators. On the one hand, the characterization refines and simplifies the characterization of path graphs due to Monma and Wei [C.L.~Monma,~and~V.K.~Wei, Intersection {G}raphs of {P}aths in a {T}ree, J. Combin. Theory Ser. B, 41:2 (1986) 141--181], which we build on, by reducing a constrained vertex coloring problem defined on the \emph{attachedness graphs} to a vertex 2-coloring problem on the same graphs. On the other hand, the characterization allows us to exhibit two exhaustive lists of obstructions to path graph membership in the form of  minimal forbidden induced/partial 2-edge colored subgraphs in each of the \emph{attachedness graphs}.
	\end{abstract}
	

	\section{Introduction}\label{introduction}
	A graph $G$ is a \emph{path graph} if there is a tree $T$ (\emph{the host tree of $G$}), a collection $\mathcal{P}$ of paths of $T$ and a bijection $\phi: V(G)\rightarrow \mathcal{P}$ such that two vertices $u$ and $v$ of $G$ are adjacent in $G$ if and only if the vertex-sets of paths $\phi(u)$ and $\phi(v)$ intersect.~Other variants of the Path/Tree intersection model are obtained by requiring edge-intersection (or even arc intersection) and by specializing the shape of $T$ (e.g.: directed, rooted). The class of path graphs is clearly closed under taking induced subgraphs. 
	Path graphs were introduced by Renz~\cite{renz} who also posed the question of characterizing them by forbidden subgraphs giving at the same a first partial answer.~The question has been fully answered only recently by L\'{e}v\^{e}que, Maffray and Preissmann~\cite{bfm}.
	
	Path graphs were first characterized by Gavril~\cite{gavril_UV_algorithm} through the notion of \emph{clique path tree} as follows (unless otherwise stated, maximal cliques are referred to as cliques, where a \emph{clique} is a set of pairwise adjacent vertices).
	
	\begin{theorem}[Gavril~\cite{gavril_UV_algorithm}]\label{thm:gavril} A graph $G$ is a path graph if and only if it possesses a \emph{clique path tree}, namely, a tree $T$ whose vertices are the cliques of $G$ with the property that the set of cliques of $G$ containing a given vertex $v$ of $G$ induces a path in $T$. 
	\end{theorem}
	The picture in the middle of Figure~\ref{fig:path_graph} shows the clique path tree of the path graph on the left. Theorem~\ref{thm:gavril2} specializes the celebrated characterization of chordal graphs, still due to Gavril~\cite{gavril1}, as those graphs possessing a \emph{clique tree} (equivalently, as the intersection graphs of a collection of subtrees in a given tree) as stated below.
	\begin{theorem}[Gavril~\cite{gavril1}]\label{thm:gavril2}
		A graph $G$ is a chordal graph if and only if it possesses a \emph{clique tree}, namely, a tree $T$ on the set of cliques of $G$ with the property that the set of cliques of $G$ containing a given vertex $v$ of $G$ induces a subtree in $T$.
	\end{theorem}
	 Notice that since a clique path tree is a particular clique tree, Theorem~\ref{thm:gavril2} also implies that path graphs are chordal. Recall that a graph is a chordal graph if it does not contain a \emph{hole} as an induced subgraph, where a hole is a chordless cycle of length at least four.~
	
	\begin{figure}[h]
		\centering
		\begin{overpic}[width=14cm,percent]{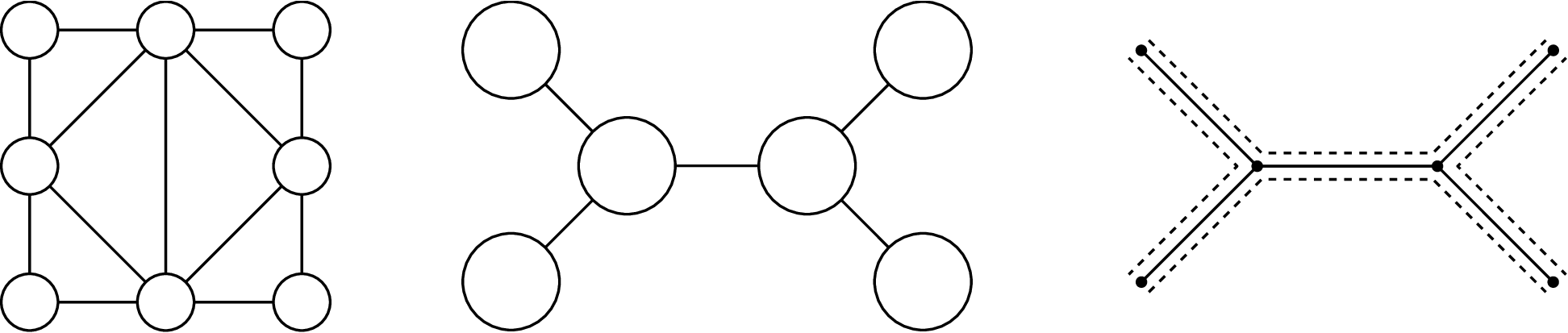}         
			
			\put(1.2,18.8){$a$}
			\put(10,18.8){$b$}
			\put(1.2,10.2){$c$}
			\put(1.2,1.4){$d$}
			\put(10,1.4){$e$}
			\put(18.5,18.8){$f$}
			\put(18.55,10.3){$g$}
			\put(18.5,1.4){$h$}
			
			\put(30.3,17.5){$a\,b\,c$}
			\put(30.3,2.7){$c\,d\,e$}
			\put(38,10.1){$b\,c\,e$}
			\put(49.3,10.1){$b\,e\,g$}
			\put(56.5,17.5){$b\,f\,g$}
			\put(56.5,2.7){$e\,g\,h$}
			
			\put(70.2,19){$p_a$}
			\put(85,12.5){$p_b$}
			\put(75,10.2){$p_c$}
			\put(70.2,2){$p_d$}
			\put(85,8.2){$p_e$}
			\put(99.7,19){$p_f$}
			\put(95,10.2){$p_g$}
			\put(99.7,2){$p_h$}
		\end{overpic}
		\caption{on the left a path graph $G$, in the center the clique path tree of $G$, on the right the host tree of $G$ and the collection $\mathcal{P}=\{p_a,\ldots,p_h\}$ that realizes $G$. Note that $p_a,p_d,p_f,p_g$ are composed by only one vertex.}
		\label{fig:path_graph}
	\end{figure}
	
	The class of path graphs is a class of graphs sandwiched between \emph{interval graphs} and \emph{chordal graphs}. A graph is an interval graph if it is the intersection graph of a family of intervals on the real line; or, equivalently, those path graphs whose host tree is a path. Interval graphs were characterized by Lekkerkerker and Boland~\cite{lekkerkerker-boland} as those chordal graphs with no \emph{asteroidal triples}, where an asteroidal triple is an independent set of three vertices such that each pair is connected by a path avoiding the neighborhood of the third vertex.
	
	Inspired by the work of Tarjan~\cite{tarjan}, Monma and Wei~\cite{mew} presented a general framework to recognize and realize intersection graphs having as intersection model all possible variants of the Path/Tree model. In particular, they characterized path graphs, \emph{directed path graphs} and \emph{rooted directed path graphs}, where the latter are subclasses of path graphs. A graph is a directed path graph if it is the intersection graph of a family of paths of a directed tree. Directed path graphs were characterized first by Panda~\cite{panda} by a list of forbidden induced subgraphs and then by Cameron, Ho\'{a}ng and L\'{e}v\^{e}que~\cite{cameron-hoang_1,cameron-hoang_2} by extending the notion of asteroidal triple. A graph is a \emph{rooted path graph} if it is the intersection graph of a family of paths of a rooted directed tree. No characterization of rooted path graphs by forbidden subgraphs or by concepts similar to asteroidal triples is currently known.
	
	Between the classes of graphs introduced above, the following inclusions hold by definition:
	\begin{equation*}
		\text{interval graphs $\subset$ rooted path graphs $\subset$ directed path graphs $\subset$ path graphs $\subset$ chordal graphs}.
	\end{equation*}
	The first recognition algorithm for path graphs was given by Gavril~\cite{gavril_UV_algorithm}, and it has $O(n^4)$ worst-case time complexity, where the input graph has $n$ vertices and $m$ edges. The two state-of-the-art fastest algorithms are due to Sch\"{a}ffer~\cite{schaffer} and Chaplick~\cite{chaplick}. Both have $O(mn)$ worst-case time complexity. The former algorithm relies on a sophisticated backtracking procedure and builds on Monma and Wei's characterization. The latter algorithm employs PQR-trees (a complex data structure). One more algorithm is proposed in \cite{dah} and claimed to run in $O(n+m)$ time. However it has only appeared as an extended abstract (see comments in [\cite{chaplick}, Section 2.1.4]).
	
	Gavril also gave the first algorithm to recognize directed path graphs~\cite{gavril_DV_algorithm}. Chaplick \emph{et al.}~\cite{chaplick-gutierrez} describe a linear algorithm able to decide whether a path graph is a directed path graph, by assuming to have the realization of the path graph as the intersection of a family of paths of a tree. This implies that the algorithms in \cite{chaplick,schaffer} can be extended to recognition algorithms for directed path graphs within the same time complexity. To the best of our knowledge these are the two state-of-the-art fastest algorithms for directed path graph recognition.
	
	\paragraph{Our Contribution}Building on Monma and Wei characterization~\cite{mew}, we give a good characterization of path graph membership within chordal graphs.~Monma and Wei characterization requires specific terminology and it is presented in detail in Section~\ref{section:mew_characterization}.~Briefly, in~\cite{mew} the input graph is decomposed recursively by \emph{clique separators} and in every decomposition step one has to solve a coloring problem (see Theorem~\ref{thm:mw} and Section~\ref{section:mew_characterization} for undefined terminology); the solution of the coloring problem is then used to build the clique path tree.  
	
	We read the coloring problem as a constrained proper vertex coloring problem in the \emph{attachedness graph} of the input graph, a graph depending on the input graph and on the clique separator (see Section~\ref{section:mew_characterization}), and, by exploiting the structure of such a graph, we reduce the constrained coloring problem to a vertex 2-coloring problem. The solution of the latter problem characterizes path graphs (see Section~\ref{section:dominance_and_antipodality}). On the one hand, this simplification implies a new polynomial-time algorithm for recognizing/realizing path graphs; on the other hand it allows us to exhibit all forbidden configurations to the property of being a path graph in the form of forbidden 2-colored subgraphs of the \emph{attachedness graphs} of the input graphs (Section~\ref{section:forbidden_subgraphs}). Main results are summarized in Theorem~\ref{cor:main} and in Theorem~\ref{cor:all}.
	
	The theoretical machinery developed in this paper is exploited algorithmically in~\cite{balzotti} where two algorithms to recognize path graphs and directed path graphs are presented. Refer to Section~\ref{sec:consequences} for a more detailed discussion on these algorithms.

	\paragraph{Notation}
	We denote by $[n]$ the interval $\{1,2,\ldots,n\}$, where $n$ is a natural number; for a subset $A$ of $V(G)$, we denote the graph induced by $A$ in $G$ by $G[A]$; a set of edges $F$ \emph{spans} a set of vertices $W$ if $W$ consists precisely of the endvertices of the edges of $F$; for a map $f:A\rightarrow B$ and $X\subseteq A$ we denote by $f(X)$ the image of $X$ under $f$, namely, $f(X)=\{f(x) \ |\ x\in X\}$. Finally, a \emph{2-edge-colored graph} is a graph whose edge-set is partitioned into two parts referred to as the \emph{edge colors} of $G$.
	
	\paragraph{Organization}
	The paper is organized as follows. In Section~\ref{section:mew_characterization} we give a detailed discussion of Monma and Wei's characterization of path graphs~\cite{mew} and we settle basic concepts and terminology. Section~\ref{section:dominance_and_antipodality} is devoted to our characterization (summarized in Theorem~\ref{cor:main} and Corollary~\ref{cor:above}), while its consequences are discussed in Section~\ref{sec:consequences}. Finally, in Section~\ref{section:forbidden_subgraphs}, we use the results to characterize paths graphs by a list of forbidden subgraphs in their attachedness graphs (see Theorem~\ref{cor:all}).

\section{Monma and Wei's characterization of path graphs}\label{section:mew_characterization}

In this section we present Monma and Wei's characterization of path graphs (Theorem~\ref{thm:mw}) which our characterization builds on. To this end we need a specific terminology.

A clique $Q$ is a \emph{clique separator} if the removal of $Q$ from $G$ disconnects $G$ into more than one connected component (without loss of generality, throughout the paper, we suppose that $G$ is connected). If graph $G$ has no clique separator, then $G$ is called \emph{atom}. In \cite{mew} it is proved that an atom is a path graph if and only if it is a chordal graph.

Given a clique separator $Q$ of a graph $G$ let $G-Q$ have $s$ connected components, $s\geq 2$ with vertex-sets $V_1,\ldots,V_s$, respectively. We define $\gamma_i=G[V_i\cup Q]$, $i=1,\ldots,s$ and $\Gamma_Q=\{\gamma_1,\ldots,\gamma_s\}$. 
A clique $K$ of a subgraph $\gamma$ of $\Gamma_Q$ is called a  \emph{relevant clique}, if $K \cap Q\not=\emptyset$ and $K\neq Q$. A \emph{neighboring subgraph} of a vertex $v\in V(G)$ is a member $\gamma\in \Gamma_Q$ such that $v$ belongs to some relevant clique $K$ of $\gamma$. For instance, in Figure~\ref{fig:example1} referring to the graph on the left, all the $\gamma_i$'s but $\gamma_5$ are neighboring subgraphs of the vertex in the north-east corner of the clique separator $Q$, while all the $\gamma_i$'s but $\gamma_2$ and $\gamma_3$ are neighboring subgraphs of the vertex in the south-west corner of $Q$. We say that two subgraphs $\gamma$ and $\gamma'$ are \emph{neighboring} if they are neighboring subgraphs of some vertex $v\in Q$; a subset $W\subseteq \Gamma_Q$ whose elements are neighboring subgraphs will be referred to as a \emph{neighboring set} (e.g, \emph{neighboring pairs}, \emph{neighboring triples} etc). Monma and Wei~\cite{mew}, defined the following binary relations on $\Gamma_Q$.

\begin{figure}[h]
\centering
\begin{overpic}[width=13cm]{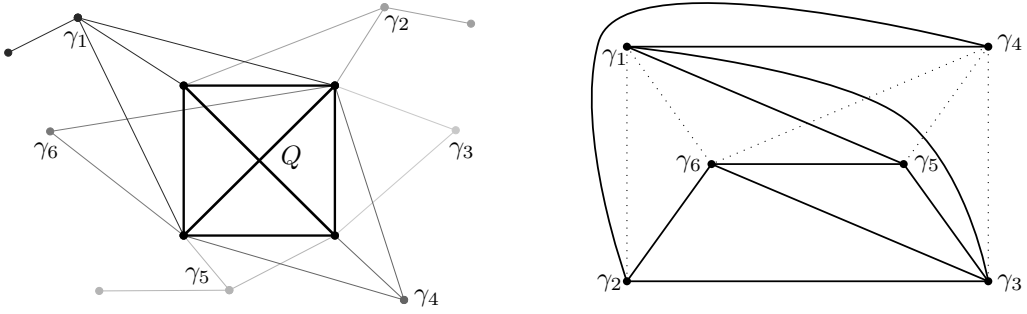}    
\put(6,26.8){$\gamma_1$} 
\put(38.5,28.5){$\gamma_2$} 
\put(45,15.5){$\gamma_3$} 
\put(41.5,0.5){$\gamma_4$} 
\put(18.3,2.7){$\gamma_5$} 
\put(3,15.5){$\gamma_6$}    
\put(28,14.5){$Q$}

\put(60.5,25){$\gamma_1$} 
\put(60,2){$\gamma_2$} 
\put(100.5,2){$\gamma_3$} 
\put(100.5,26.5){$\gamma_4$} 
\put(92.1,14.2){$\gamma_5$} 
\put(68,14.2){$\gamma_6$}    
\end{overpic}
     \caption{A graph $G$ (on the right) with a unique clique separator $Q$. Connected components separated by $Q$ are represented with different grays The $Q$-attachedness graph of $G$ is on the left. Remark that graph G is not a path graph.}
\label{fig:example1}
\end{figure}
\mybreak
\emph{Attachedness}, denoted by $\Join$ and defined by $\gamma\Join \gamma'$ if and only if there is a relevant clique $K$ of $\gamma$ and a relevant clique $K'$ of $\gamma'$ such that $K\cap K'\cap Q\neq\emptyset$. In particular, $\gamma$ and $\gamma'$ are neighboring subgraphs of each vertex $v\in K\cap K'\cap Q$. 
\mybreak
\noindent \emph{Dominance}, denoted by $\leq$ and defined by $\gamma\leq \gamma'$ if and only if $\gamma\Join \gamma'$ and for each relevant clique $K'$ of $\gamma'$ either $K\cap Q\subseteq K'\cap Q$ for each relevant clique $K$ of $\gamma$ or $K\cap K'\cap Q=\emptyset$ for each relevant clique $K$ of $\gamma$. Pairs of $\leq$-comparable subgraphs of the graph $G$ on the left of Figure~\ref{fig:example1} are joined by a dotted edge on the right of the same figure.
\mybreak
\noindent \emph{Antipodality}, denoted by $\leftrightarrow$ and defined by $\gamma \leftrightarrow\gamma'$ if and only if there are relevant cliques $K$ of $\gamma$ and $K'$ of $\gamma'$ such that $K\cap K'\cap Q\not=\emptyset$ and $K\cap Q$ and $K'\cap Q$ are inclusion-wise incomparable. Still referring to Figure~\ref{fig:example1}, pairs of antipodal subgraphs of $G$ are joined by a solid edge on the right of the same figure.
\mybreak
The following lemma is implied by definitions of $\leq$ and $\leftrightarrow$.

\begin{lemma}\label{lemma:elenco_trivial_1}
Let $Q$ be a clique separator of $G$ and let $\gamma,\gamma'\in\Gamma_Q$, the following hold:
\begin{enumerate}[label={\rm (\arabic*)}]\itemsep0em
\item\label{item:neighboring_leq} $\gamma\leq\gamma'\Rightarrow \gamma$ and $\gamma$ are neighboring of $v$, for all $v\in V(\gamma)\cap Q$
\item\label{item:neighboring_ant} $\gamma\leftrightarrow\gamma'\Rightarrow \gamma$ and $\gamma$ are neighboring of $v$, for all $v\in V(\gamma)\cap V(\gamma')\cap Q$.
\end{enumerate}
\end{lemma}
\noindent
Antipodality and dominance relations are disjoint binary relations on $\Gamma_Q$ whose union is the relation $\Join$. Therefore $(\gamma\leq \gamma'$, $\gamma'\leq \gamma$ or $\gamma\leftrightarrow \gamma')$ if and only if $(\gamma\Join\gamma')$. Both $\Join$ and $\leftrightarrow$ are symmetric and only $\leftrightarrow$ is irreflexive. Hence, after neglecting reflexive pairs, $(\Gamma_Q,\leftrightarrow)$,  $(\Gamma_Q,\Join)$ are simple undirected graphs on $\Gamma_Q$ referred to as, respectively, the \emph{$Q$-antipodality} and the \emph{$Q$-attachedness graph of $G$}. The edges of the $Q$-antipodality graph of $G$ are called \emph{antipodal edges} while those edges of the $Q$-attachedness graph of $G$ which are not antipodal edges, are called \emph{dominance edges}. The \emph{$Q$-dominance} graph of $G$ is the graph on $\Gamma_Q$ having as edges the dominance edges (i.e., the complement of $(\Gamma_Q,\leftrightarrow)$ in $(\Gamma_Q,\Join)$). Hence the edge-sets of the $Q$-antipodality and the $Q$-dominance graphs of $G$ partition the edge-set of the $Q$-attachedness graph of $G$, and the latter is naturally 2-edge colored by the antipodality edges and by the dominance edges. We adopt the pictorial convention to represent antipodality edges by thin lines and dominance edges by dotted lines. In Figure~\ref{fig:example1}, the graph $G$ on the left ($G$ is not a path graph) has a unique clique separator $Q$. Its $Q$-attachedness graph is shown on the right of the same figure. Due to transitivity of $\leq$, the following diagrams (which are triangles in the $Q$-attachedness graph of $G$) represent all possible cases involving three pairwise attached elements of $\Gamma_Q$.
\begin{equation}\label{diagram:triangles}
	\begin{gathered}
		\xymatrix@R-1pc@C=3mm{
			&  *[o]+<5pt>{\gamma''} &\\
			& & *[o]+<5pt>{\gamma'}\ar@{-}[ul]  \\
			*[o]+<5pt>{\gamma}\ar@{-}[uur]\ar@{-}[urr] & &\\ & (a) &}\qquad
		\xymatrix@R-1pc@C=3mm{
			&  *[o]+<5pt>{\gamma''} &\\
			& & *[o]+<5pt>{\gamma'}\ar@{.}[ul]  \\
			*[o]+<5pt>{\gamma}\ar@{.}[uur]\ar@{.}[urr] & &\\ & (b) &}\qquad	
		\xymatrix@R-1pc@C=3mm{
			&  *[o]+<5pt>{\gamma''} &\\
			& & *[o]+<5pt>{\gamma'}\ar@{-}[ul]  \\
			*[o]+<5pt>{\gamma}\ar@{.}[uur]\ar@{-}[urr] & &\\ & (c) &}\qquad
		\xymatrix@R-1pc@C=3mm{
			&  *[o]+<5pt>{\gamma''} &\\
			& & *[o]+<5pt>{\gamma'}\ar@{.}[ul]  \\
			*[o]+<5pt>{\gamma}\ar@{.}[uur]\ar@{-}[urr] & &\\ & (d) &}\qquad
		\xymatrix@R-1pc@C=3mm{
			&  *[o]+<5pt>{\gamma''} &\\
			& & *[o]+<5pt>{\gamma'}\ar@{-}[ul]  \\
			*[o]+<5pt>{\gamma}\ar@{.}[uur]\ar@{.}[urr] & &\\ & (e) &}
	\end{gathered}
\end{equation} 
Notice that by Lemma~\ref{lemma:elenco_trivial_1}, while any of the  triangles (b),(c),(d),(e) in \eqref{diagram:triangles} is induced by a neighboring triple, the same does not occur for the triangle in (a), namely, a triangle in the $Q$-antipodality graph is not necessarily induced by a neighboring triple. We are now ready to state Monma and Wei's characterization of path graphs which builds on the following notion.

\begin{definition}\label{def:strong_colorability}
Let $Q$ be a clique separator of $G$, we say that $G$ is \emph{strong $Q$-colorable} if there exists $f:\Gamma_Q\rightarrow [s]$ such that:
\begin{enumerate}[label={\rm (\arabic*)}]\itemsep0em
\item\label{com:mw_i} if $\gamma\leftrightarrow \gamma'$, then $f(\gamma)\not=f(\gamma')$;
\item\label{com:mw_ii} if $\{\gamma,\gamma',\gamma''\}$ is neighboring triple, then $|f(\{\gamma,\gamma',\gamma''\})|\leq 2$.
\end{enumerate}
\end{definition}
We refer to a coloring $f$ satisfying the conditions of Definition~\ref{def:strong_colorability} as a \emph{strong $Q$-coloring}. We use the term ``strong" because in Section~\ref{section:dominance_and_antipodality} we introduce a weaker notion of coloring and we prove that they are equivalent. 
\begin{remark}\label{rem:pro_coloring}
By condition~\ref{com:mw_i} in the definition above, any strong $Q$-coloring of $G$ is a proper coloring of the $Q$-antipodality graph of $G$. Condition~\ref{com:mw_ii} requires that all neighboring triples are 2-colored under this proper coloring.  
\end{remark}
\begin{remark}
Dominance is a reflexive and transitive relation. Hence $(\Gamma_Q,\leq)$ is a preorder. We assume that such a preorder is in fact a partial order. The latter assumption is not restrictive as showed implicitly in Sch\"{a}ffer~\cite{schaffer} and explicitly as follows. Let $\sim$ be the equivalence relation on $\Gamma_{Q}$ defined by  $\gamma\sim\gamma'\Leftrightarrow\left(\gamma\leq\gamma'\wedge \gamma'\leq\gamma\right)$, namely, $\sim$ is the standard equivalence relation associated with a preorder. It is easily checked that being antipodal, neighboring and strong $Q$-colorable, are properties that all \emph{pass to the quotient} $\Gamma_Q/\sim$. Hence, unless otherwise stated, we assume that $(\Gamma_Q,\leq)$ is a partial order for every clique separator $Q$ of $G$. In other words, we assume that $\Gamma_Q=\Gamma_Q/\sim$.
\end{remark}
The reason for studying strong $Q$-coloring relies on the fact that the two conditions in Definition~\ref{def:strong_colorability} are necessary for membership in path graphs. Indeed, suppose $G$ is a path graph and let $T$ be a clique path tree of $G$ (refer to Figure~\ref{fig:path_graph}). The removal of clique separator $Q$ from $G$ disconnects $G$ in more connected components, but it also disconnects $T$ in more subtrees, the branches of $T$. Each $\gamma\in \Gamma_Q$ is thus a path graph having its own clique tree (a subtree of $T$) which lies in exactly one branch of $T$. Coloring $f$ associates connected components of $G$ with the branches of $T$. The first condition implies that two antipodal connected components map into two distinct branches, for if not there is some $v\in Q$ such that the set of cliques of $G$ that contain $v$ is not connected in $T$. The second condition requires that all connected components that contain $v$ need to be in at most two distinct branches, for if not the set of cliques of $G$ that contain $v$ induces a graph in $T$ with a vertex of degree at least three. Summarizing, both conditions imply that the set of cliques that contain $v$, for all $v\in Q$, induces in $T$ a connected graph in which every vertex has degree at most two, i.e., a path, as required by Theorem~\ref{thm:gavril}. Monma and Wei's characterization shows that these conditions are also sufficient for path graphs membership. 
\begin{theorem}[Monma and Wei~\cite{mew}]\label{thm:mw}
A chordal graph $G$ is a path graph if and only if $G$ is an atom or for a clique separator $Q$ each graph $\gamma\in \Gamma_Q$ is a path graph and $G$ is strong $Q$-colorable.
\end{theorem}
The following non recursive restatement of Theorem~\ref{thm:mw} is more well suited for our purposes. We recall that a graph with no clique separator (i.e., an atom) is a path graph if and only if it is chordal.
\begin{corollary}\label{cor:local_mew}
A chordal graph $G$ is a path graph if and only if $G$ is strong $Q$-colorable, for all  clique separators $Q$ of $G$.
\end{corollary}
By Corollary~\ref{cor:local_mew}, deciding whether a graph $G$ is a path graph is tantamount to decide whether $G$ is strong $Q$-colorable for each separator $Q$. It is thus natural to wonder whether there are obstructions to strong $Q$-colorability and, in case, how do such obstructions look like in attachedness graph of $G$. One of such obstruction is easily recognized (see \cite{mew}): let $\{\gamma,\gamma',\gamma''\}\subseteq \Gamma_Q$ be a neighboring triple and suppose that $\gamma$, $\gamma'$ and $\gamma''$ are pairwise antipodal. Hence $\{\gamma,\gamma',\gamma''\}$ induces a triangle in the $Q$-antipodal graph of $G$ though not all triangles of the $Q$-antipodality graph correspond to neighboring triples (in Section~\ref{section:forbidden_subgraphs}, we give a detailed account of this phenomenon). We refer any such triple to as a \emph{full antipodal triple}. It is clear that if $\Gamma_Q$ contains a full antipodal triple, then $G$  is not strong $Q$-colorable because the two conditions in Definition~\ref{def:strong_colorability} cannot be both satisfied. 
For later reference we formalize this easy fact in a lemma.
\begin{lemma}\label{lemma:nofull}
Let $Q$ be a clique separator of $G$. If $G$ is strong $Q$-colorable, then $\Gamma_Q$ contains no full antipodal triple.
\end{lemma}	  
\noindent
Besides full antipodal triples there are many other obstructions and in this paper we exhibit the exhaustive list of such obstructions in the form of forbidden subgraphs of the $Q$-attachedness graph of $G$. The result is achieved by means of a new characterization of strong $Q$-coloring (Section~\ref{section:dominance_and_antipodality}) which is actually a weakening of the original notion. Our characterization yields a reduction of strong $Q$-colorability of $G$ to 2-colorability of the members of a certain partition of $\Gamma_Q$ providing at the same time the so-called \emph{good characterization} for path graph membership within chordal graphs, namely, it proves that path graph membership of chordal graphs is in $\text{NP}\cap\text{CoNP}$ without resorting to existing polynomial-time algorithms (Theorem~\ref{cor:main} and Corollary~\ref{cor:above}).

\section{Characterizing path graphs through weak colorings}\label{section:dominance_and_antipodality}
By Remark~\ref{rem:pro_coloring}, a strong $Q$-coloring is a proper coloring of the $Q$-antipodal graph of $G$ fulfilling the condition that neighboring triples are 2-colored. Full antipodal triples are special neighboring triples and play a distinguished role in characterizing path graphs. This fact is somehow hidden in Monma and Wei' characterization~\cite{mew} because the absence of full antipodal triples in $\Gamma_Q$ is a direct consequence of Theorem~\ref{thm:mw} (recall Lemma~\ref{lemma:nofull}). Nonetheless, as we show, by singling out the property of being full antipodal triple free, entails a stiff structure of strong $Q$-colorings of $G$ (Theorem~\ref{thm:canonical} via Lemma~\ref{remark:D_i_is_a_partition} and Lemma~\ref{lemma:elenco_trivial_2}) and an even stiffer structure of the $Q$-antipodality graph of $G$ (see Corollary~\ref{cor:above}). Such a structure allows to easily find \emph{weak $Q$-colorings}---simpler and more refined proper colorings of the $Q$-antipodality graph of $G$ (see Definition~\ref{def:partial_and_weak})---, which we prove to be equivalent to strong-$Q$-colorings (Theorem~\ref{lemma:characterization_1_only_if}). These are the bases of our characterization.        

For a clique separator $Q$ of $G$, let $\Upper_Q=\{u\in\Gamma_Q \,|\, u\not\leq\gamma, \text{ for all } \gamma\in\Gamma_Q\}$ the set of \emph{upper bounds} of $\Gamma_Q$ with respect to $\leq$.

From now on we fix $(u_1,u_2,\ldots,u_\ell)$ an ordering of $\Upper_Q$. For all $i,j\in[\ell]$ and $i<j$ we define

\begin{equation}\label{dgamma}	
	D_{i}^Q=\{\gamma\in \Gamma_Q \ |\ \gamma\leq u_i  \text{ and } \gamma\nleq u_j,\,\,\forall j\in[\ell]\setminus \{i\}\},
\end{equation}
\begin{equation}\label{dgammagamma'}
	D_{i,j}^Q=\{\gamma\in\Gamma_Q \ |\ \gamma\leq u_i,\gamma\leq u_j  \text{ and } \gamma\nleq u_k,\,\,\forall k\in[\ell]\setminus \{i,j\}\},
\end{equation}	
namely, $D_i$ consists of the elements of $\Gamma_Q$ dominated only by $u_i$ and no other upper bound, while $D_{i,j}$ consists of those elements of $\Gamma_Q$ dominated only by $u_i$ and $u_j$ and no other upper bound. Referring to Figure~\ref{fig:example1}, $\Upper_Q=\{\gamma_1,\gamma_4\}$, if we fix the ordering $(u_1,u_2)=(\gamma_1,\gamma_4)$, then $D_1=\{\gamma_1,\gamma_2\}$, $D_2=\{\gamma_3,\gamma_4\}$ and $D_{1,2}=\{\gamma_6\}$. Moreover let
\begin{equation*}\label{dstorto}
	\mathcal{D}^Q=\Big\{D^Q_i \,|\, i\in[\ell]\Big\}\cup\Big\{D^Q_{i,j} \,|\, i,j\in[\ell],i<j\Big\}.
\end{equation*}	
and refer to $\mathcal{D}^Q$ as the \emph{$Q$-skeleton} of $G$. If no confusion arises, we omit the superscript $Q$ in the notation above.

\begin{lemma}\label{remark:D_i_is_a_partition}
	Let $Q$ be a clique separator of $G$ and let $\mathcal{D}$ be the $Q$-skeleton of $G$. If $\Upper_Q$ contains no full antipodal triple, then $\mathcal{D}$ is a partition of $\Gamma_Q$.
\end{lemma}
\begin{proof}
	The members of $\mathcal{D}$ are pairwise disjoint by construction. Hence $\mathcal{D}$ is a partition of $\Gamma_Q$ if and only if $\cup_{D\in\mathcal{D}}D=\Gamma_Q$ and this happens if and only if every $\gamma\in \Gamma_Q\setminus \Upper_Q$ has at most two upper bounds in $\Upper_Q$. If some $\gamma\in \Gamma_Q\setminus \Upper_Q$ had three upper bounds in $\Upper_Q$, then by a repeated application of Lemma~\ref{lemma:elenco_trivial_1}, the triple formed by these upper bounds would be a neighboring triple consisting of three mutually $\leq$-incomparable elements and hence a full antipodal triple contained in $\Upper_Q$.
\end{proof}

\begin{lemma}\label{lemma:elenco_trivial_2}
	Let $Q$ be a clique separator of $G$. Let $i,\,j\in[\ell]$, $i<j$. If $\Upper_Q$ contains no full antipodal triple, then the following hold:
	\begin{enumerate}[label={\rm(\arabic*)}]\itemsep0em
		\item\label{item:antipodal_k_k'_1} $\gamma\leftrightarrow\gamma'$, $\gamma\in D_{i,j}$ and $\gamma'\not\in D_{i,j}$ $\Rightarrow$ $\gamma'\in D_i\cup D_j$,
		\item\label{item:antipodal_D_i_Upper} $\gamma\leftrightarrow\gamma'$, $\gamma\in D_i$ and $\gamma'\not\in D_i$ $\Rightarrow\gamma\leftrightarrow u_k$ for any $k\neq i$ such that $\gamma'\leq u_k$,
		\item\label{item:antipodal_square} $\gamma\leftrightarrow\gamma'$, $\gamma\in D_i$ and $\gamma'\in D_j$ $\Rightarrow\gamma\leftrightarrow u_j$ and $\gamma'\leftrightarrow u_i$.
	\end{enumerate}
\end{lemma}
\begin{proof}
	If \ref{item:antipodal_k_k'_1} fails to hold, then there exists $k\in[\ell]\setminus\{i,j\}$ such that $\gamma'\leq u_k$. Hence Lemma~\ref{lemma:elenco_trivial_1} implies that $\{u_i,u_j,u_k\}$ is a full antipodal triple, contradicting that $\Upper_Q$ contains no such a triple. To prove \ref{item:antipodal_D_i_Upper} observe that by Lemma~\ref{remark:D_i_is_a_partition}, if $\gamma'\not\in D_i$ then $\gamma'\leq u_k$ for some $k\neq i$. By Lemma~\ref{lemma:elenco_trivial_1}, $\gamma$ and $u_k$ are neighboring. Hence $\gamma$ and $u_k$ are attached with $u_k\in \Upper_Q$. Thus either $\gamma\leftrightarrow u_k$ or $\gamma\leq u_k$. The latter relation cannot hold by the definition of $D_i$. We therefore conclude that $\gamma\leftrightarrow u_k$ as claimed. Finally, \ref{item:antipodal_square} is a straightforward consequence of \ref{item:antipodal_D_i_Upper}.
\end{proof}
By the two previous lemmas, the $Q$-antipodality graph is equipped with an extra structure: the $Q$-skeleton of $G$. By looking at how strong $Q$-colorability acts within and between the members of the $Q$-skeleton of $G$ provides a significant (and easily checkable) refinement of the two conditions in Definition~\ref{def:strong_colorability}.

By Lemma~\ref{remark:D_i_is_a_partition} the $Q$-skeleton $\mathcal{D}$ of $G$ partitions the vertex set of its $Q$-antipodality graph $H$. This partition induces a partition of the edges of $H$ into two set:
\begin{itemize}\itemsep0em
	\item the set of \emph{cross edges}, namely, those edges of $H$ having their ends in different members of $\mathcal{D}$, 
	\item the set of \emph{intra edges}, namely, those edges having both ends in the same member of  $\mathcal{D}$.  
\end{itemize}
Let $H^\cross$ be the subgraph of $H$ spanned by the cross edges. Analogously, let $H^\intra$ be the subgraph of $H$ spanned by the intra edges. Observe that 
$$V(H^\cross)=\text{$\{\gamma\in\Gamma_Q \,|\, \gamma\in D$, for some $D\in\mathcal{D}$, and $\gamma\leftrightarrow\gamma'$, for some $\gamma'\not\in D\}$}.$$

\begin{theorem}\label{thm:canonical}
	Let $Q$ be a clique separator of $G$ and let $\ell=|\Upper_Q|$. If $G$ is strong $Q$-colorable, then there exists a strong $Q$-coloring $f:\Gamma_Q\rightarrow [\ell+1]$ satisfying the following:
	\begin{enumerate}[label={\rm(\alph*)}]\itemsep0em
		\item\label{item:a} for all $i\in[\ell]$, $f(u_i)=i$, 
		\item\label{item:b} for all $i\in[\ell]$, for all $\gamma\in D_i$, $f(\gamma)\in\{i,\ell+1\}$,
		\item\label{item:c} for all $i,\,j\in[\ell]$, $i<j$, for all $\gamma\in D_{i,j}$, $f(\gamma)\in\{i,j\}$,
		\item\label{item:d} for all $i\in[\ell]$, for all $\gamma\in D_i$ if $\exists u\in\Upper_Q$ such that $\gamma\leftrightarrow u$, then $f(\gamma)=i$,
		\item\label{item:e} for all $i,\,j\in[\ell]$, $i<j$, for $k\in\{i,j\}$, and for all $\gamma\in D_{i,j}$, if $\exists\gamma'\in D_{k}$ such that $\gamma\leftrightarrow\gamma'$, then $f(\gamma)=\{i,j\}\setminus\{k\}$,
		\item\label{item:f} for all $D\in\mathcal{D}$, for all $\gamma,\gamma\in D$ such that $\gamma\leftrightarrow\gamma'$, $f(\gamma)\neq f(\gamma')$,
	\end{enumerate}
	where $\{D_i \ |\ i\in [\ell]\}\cup\{D_{i,j} \ |\ i,\,j\in [\ell], i<j\}$ is the $Q$-skeleton $\mathcal{D}$ of $G$. 
\end{theorem}
\begin{proof} Let $H$ be the $Q$-antipodality graph of $G$ and let $g$ be any strong $Q$-coloring of $G$. Since $g$ is a strong $Q$-coloring, it is, in particular, a proper coloring of $H$. Hence \ref{item:f} is trivially satisfied by $g$ and we have only to prove the theorem for the first five conditions. To this end, we first prove that under the hypothesis of strong $Q$-colorability, $g$ satisfies a number of properties, and then we use such properties to modify $g$ into a strong $Q$-coloring $f$ which satisfies the first five conditions of the theorem. 
	
	For $i\in [\ell]$ let $A_i=\{\gamma\in D_i\ |\ g(\gamma)\not=g(u_i)\}$, $B_i=D_i\setminus A_i$, and $A_{\ell+1}=\cup_i^\ell A_i$. Notice that $u_i\in B_i$ hence $B_i\not=\emptyset$. Clearly $A_i$ and $B_i$ form a bipartition of $D_i$ and the elements of $B_i$ all have the same color $g(u_i)$. Hence the elements of $B_i$ are pairwise non antipodal ($B_i$ is thus an independent set in $H$). We claim that 
	\begin{equation}\label{eq:stable}
		\text{$\forall i\in [\ell]$ it holds that $A_i$ is an independent set of $H$, and $\gamma\not\Join\gamma'$ for each $\gamma\in A_i$, $\gamma'\in\Gamma_Q\setminus D_i$.}
	\end{equation}	
	To prove that $A_i$ is an independent set of $H$ for all $i\in [\ell]$ observe that if $\gamma,\,\gamma'\in A_i$ were antipodal, then $\{\gamma,\gamma',u_i\}$ would be a neighboring triple by Lemma~\ref{lemma:elenco_trivial_1}. Now, $g(\gamma)\neq g(\gamma')$ (because $g$ is a proper coloring of $H$) and $g(u_i)\not\in \{g(\gamma),g(\gamma)\}$ (by the definition of $A_i$). However this is impossible because neighboring triples have to be 2-colored under $g$ by Definition~\ref{def:strong_colorability}.~\ref{com:mw_ii}. Let us prove the second part of \ref{eq:stable}. Suppose by contradiction that $\gamma\Join \gamma'$ for some $\gamma\in A_i$,  $\gamma'\in \Gamma_Q\setminus D_i$, and $i,\,j\in [\ell]$, $i<j$. Since $\gamma'\not\in D_i$, there is an upper bound $u'$ of $\gamma'$ such that $u'\neq u_i$. By transitivity of $\leq$ one must have $\gamma\not\leq\gamma'$, otherwise $\gamma\leq u'$ would hold true, contradicting that $\gamma\in D_i$. Hence $\gamma\Join\gamma'$ implies either $\gamma\leftrightarrow\gamma'$ or $\gamma'\leq\gamma$. In either case, $\gamma\leftrightarrow u'$ and $u_i\leftrightarrow u'$ both hold by Lemma~\ref{lemma:elenco_trivial_2} and the definition of $D_i$. It follows that Diagram~\ref{diagram:triangles}.(c) applies with $u'$ in place of $\gamma'$ and $u_i$ in place of $\gamma''$. Thus $\{\gamma,u',u_i\}$ must be 2-colored. However this is impossible. Indeed $g(\gamma)$ and $g(u_i)$ should coincide, because $\gamma$ and $u_i$ are both antipodal to $u'$ while $g(u_i)\not=g(\gamma)$ because $\gamma\in A_i$. 
	We therefore conclude that~\eqref{eq:stable} holds. Notice, in passing, that~\eqref{eq:stable} implies $A_{\ell+1}\cap V(H^\cross)=\emptyset$. By~\eqref{eq:stable} it also follows that
	\begin{equation}\label{eq:stable2}
		\text{if $\Lambda\cap A_{\ell+1}\neq \emptyset$ for some neighboring triple $\Lambda\subseteq \Gamma_Q$, then $\Lambda\subseteq D_i$ for some $i\in [\ell]$.}
	\end{equation}	
	For if not, there are $i\in [\ell]$, $\gamma\in \Lambda\cap D_i$ and $\gamma'\in \Lambda\setminus D_i$ such that $\gamma\Join\gamma'$ (because the elements of every neighboring set are pairwise attached), contradicting~\eqref{eq:stable}.  
	
	Next we prove that	
	\begin{equation}\label{eq:stable3}
		\text{for all $i,j\in[\ell]$, $i<j$, and for all $\gamma\in D_{i,j}$, $g(\gamma)\in\{g(u_i),g(u_j)\}$}
	\end{equation}	
	To prove~\eqref{eq:stable3} let $\gamma\in D_{i,j}$ and observe that $u_i$ and $u_j$ are the unique upper bounds of $\gamma$ because of Lemma~\ref{remark:D_i_is_a_partition} and the definition of $D_{i,j}$ (see \ref{dgammagamma'}). By Lemma~\ref{lemma:elenco_trivial_1}, $\gamma$, $u_i$ and $u_j$ are pairwise neighboring. Moreover, $u_i$ and $u_j$  are $\leq$-incomparable. Hence Diagram~\ref{diagram:triangles}.(e) applies with $u_i$ and $u_j$ in place of $\gamma'$ and $\gamma''$. Since $g$ is a strong $Q$-coloring, it follows that $g(u_i)\not=g(u_j)$ and  $|g(\{\gamma,u_i,u_j\})|\leq 2$. We conclude that $g(\gamma)\in \{g(u_i),g(u_j)\}$ as required.
	
	By~\eqref{eq:stable},~\eqref{eq:stable3}, and the definitions of $A_{\ell+1}$, and $B_i$, it follows that for each $\gamma\in \Gamma_Q\setminus A_{\ell+1}$ there is a unique $u\in\Upper_Q$ such that $g(\gamma)=g(u)$. Denote such an upper bound by $\nu^g(\gamma)$. Hence the map $\nu^g: \Gamma_Q\setminus A_{\ell+1}\rightarrow \Upper_Q$ is well defined and
	\begin{itemize}\itemsep0em
		\item $\nu^g(\gamma)=u_i$, for all $\gamma\in B_i$ and $i\in[\ell]$, 
		\item $\nu^g(\gamma)\in\{u_i,u_j\}$, for all $\gamma\in D_{i,j}$ and $i,j\in[\ell]$, $i<j$. 
	\end{itemize}
	Let $\tau: \Upper\rightarrow [\ell]$ be the unique bijection such that $\tau(u_i)=i$. We claim that the map $f:\Gamma_Q\rightarrow [\ell+1]$ defined by
	\begin{equation*}
		f(\gamma)=\left\{
		\begin{array}{ll}
			\ell+1 & \text{if $\gamma\in A_{\ell+1}$}\\
			\tau(\nu^g(\gamma)) & \text{otherwise},
		\end{array}
		\right.
	\end{equation*}
	is strong $Q$-coloring of $G$. In the first place, $f$ is a proper coloring of $H$. Indeed $A_{\ell+1}$ is an independent set of $H$; moreover, if $\gamma\leftrightarrow\gamma'$ for some two $\gamma,\gamma'\in \Gamma_Q\setminus A_{\ell+1}$, then $g(\gamma)\neq g(\gamma')\Rightarrow f(\gamma)\neq f(\gamma')$ because $\gamma$ and $\gamma'$ get their respective colors from different upper bounds. In the second place, it holds that $|f(\Lambda)|\leq 2$ for every neighboring triple $\Lambda$ contained in $\Gamma_Q$. Indeed, if $\Lambda\cap A_{\ell+1}\neq \emptyset$, then there exists some $i\in [\ell]$ such that $\Lambda\subseteq D_i$ by~\eqref{eq:stable2}. Hence $f(\Lambda)\subseteq \{\ell+1,i\}$ and we are done. Otherwise, if $\Lambda\cap A_{\ell+1}=\emptyset$, then $|f(\Lambda)|=|\{\nu^g(\gamma) \ |\ \gamma\in \Lambda\}|$ (because $\tau$ is injective).~By a repeated application of Lemma~\ref{lemma:elenco_trivial_1}, $\{\nu^g(\gamma) \ |\ \gamma\in \Lambda\}$ is a neighboring set contained in $\Upper_Q$. Hence, if $|f(\Lambda)|=3$, then $\{\nu^g(\gamma) \ |\ \gamma\in \Lambda\}$  would be a full antipodal triple contained in $\Gamma_Q$, contradicting Lemma~\ref{lemma:nofull}. We conclude that $f$ is a strong $Q$-coloring which, by construction, satisfies \ref{item:a},~\ref{item:b},~\ref{item:c}. 
	
	It remains to prove that $f$ satisfies \ref{item:d}, and~\ref{item:e} as well. As for \ref{item:d}, observe that if $\gamma$ qualifies for the condition, then $\gamma\in D_i\cap V(H^\cross)\subseteq D_i\setminus A_{\ell+1}$ by~\eqref{eq:stable}. Hence $f(\gamma)=i$ by \ref{item:b}. As for \ref{item:e}, observe that if $\gamma$ qualifies for the condition, then $\gamma\in D_{i,j}$ and there exists $\gamma'\in D_k\setminus A_{\ell+1}$ for $k\in\{i,j\}$ such that $\gamma\rightarrow \gamma'$. Now $f(\gamma')=k$ by \ref{item:b}, and $f(\gamma)\in\{i,j\}$ by \ref{item:c}. Hence $f(\gamma)\in\{i,j\}\setminus\{k\}$ because $f$ is (in particular) a proper coloring of $H$. The proof is now completed.\end{proof}

\begin{remark}\label{rem:tau}
	Bijection $\tau$ in the proof above can actually be chosen arbitrarily but \ref{item:a}, \ref{item:b}, and~\ref{item:c} have to be modified accordingly: \rm {(a') the restriction of $f$ to $\Upper_Q$ uses exactly $\ell$ colors}; \rm {(b')
		for all $i\in[\ell]$, for all $\gamma\in D_i$, $f(\gamma)\in\{f(u_i),\ell+1\}$}, and \rm {(c') for all $i,\,j\in[\ell],i<j$, for all $\gamma\in D_{i,j}$, $f(\gamma)\in\{f(u_i),f(u_j)\}$}.  
\end{remark}
Theorem~\ref{thm:canonical} sharpens the two defining conditions of strong $Q$-colorability (Definition~\ref{def:strong_colorability}) into a number of easily checkable conditions. 

Conditions~\ref{item:d} and \ref{item:e} rule the behavior of the strong $Q$-colorings $f$ on the edges of $H^\cross$ while conditions \ref{item:b}, \ref{item:c}, and \ref{item:f} rule the behavior of $f$ on the edges of $H^\intra$. Condition~\ref{item:a} prescribes the number of colors globally used by $f$.

If $g$ is not assumed to be a strong $Q$-coloring of $G$, then all the six conditions are independent from each other. All such  conditions are clearly necessary for the strong $Q$-colorability of $G$. We prove in Theorem~\ref{lemma:characterization_1_only_if} that these conditions are also sufficient for strong $Q$-colorability yielding the new characterization. To avoid a boring steady mention to the six conditions of Theorem~\ref{thm:canonical} we give the following definition.
\begin{definition}\label{def:partial_and_weak}
	Let $Q$ be a clique separator of $G$ and let $\ell=|\Upper_Q|$. If $\Upper_Q$ contains no full antipodal triple, then we say that $G$ is \emph{weak $Q$-colorable} if there exists $f:\Gamma_Q\rightarrow [\ell+1]$ satisfying all of the conditions in Theorem~\ref{thm:canonical}. If such an $f$ exists, then we refer to it as a \emph{weak $Q$-coloring of $G$}.
\end{definition}
We defer a justification for the adjective ``weak'' after Theorem~\ref{lemma:characterization_1_only_if}; for the moment we only notice that the first five conditions in the definition of weak $Q$-coloring, uniquely determine a proper coloring $f^\cross$ of $H^\cross$ which uses exactly $\ell$-colors. 

Indeed, if $\gamma$ and $\gamma'$ are adjacent vertices in $H^\cross$, then $\gamma\leftrightarrow \gamma'$ with $\gamma\in D$ and $\gamma'\in D'$ for some two members $D$ and $D'$ of the $Q$-skeleton of $G$. Hence, there are $i,\,j\in [\ell]$ such that either $D=D_i$ and $D'=D_j$, or $D=D_i$ and $D'=D_{i,j}$, or $D=D_j$ and $D'=D_{i,j}$. By Lemma~\ref{lemma:elenco_trivial_2}.~\ref{item:antipodal_k_k'_1} there are no other possibilities. Moreover, if $D=D_i$ and $D'=D_j$, then  $\gamma\leftrightarrow u_j$ and $\gamma'\leftrightarrow u_i$ both hold by Lemma~\ref{lemma:elenco_trivial_2}.~\ref{item:antipodal_square}. Hence $f(\gamma)\neq f(\gamma')$ in each of the three cases: in the first case by \ref{item:b} and \ref{item:d} while in the remaining cases by \ref{item:b}, \ref{item:c} and \ref{item:e}. Remark that $f$ is a map on $\Gamma_Q$. Hence there cannot exist $\gamma\in D_{i,j}$, $\gamma'\in D_i$, and $\gamma''\in D_j$ such that $\gamma\leftrightarrow \gamma'$ and $\gamma\leftrightarrow \gamma''$ both hold, otherwise $f$ could not have been defined. 
Uniqueness of $f^\cross$ is entailed by conditions \ref{item:d}, and \ref{item:f}. 
Conditions~\ref{item:b},\ref{item:c}, and \ref{item:f} on the other hand, determines a proper 2-coloring in the subgraph induced in $H$ by each member $D$ of $Q$-skeleton. Hence, the same conditions determine a proper 2-coloring of $H^\intra$ because such graph is precisely the union of the $H[D]$'s as $D$ runs in $\mathcal{D}$. Clearly, altogether, the six conditions define a proper coloring of $H$ because such a coloring is defined within and between components. Conversely, if $f$ is a proper coloring of $H$ which satisfies the first four conditions of Definition~\ref{def:partial_and_weak}, then $f$ is a weak $Q$-coloring: condition~\ref{item:f} is trivially satisfied, while condition \ref{item:e} must necessarily hold given the conditions~\ref{item:a}, \ref{item:b},\ref{item:c}, \ref{item:d}, and the fact that $f$ is a proper coloring. Therefore we have the following:
\begin{theorem}\label{thm:restricted_col}
	Let $Q$ be a clique separator of a graph $G$, and $H$ be the $Q$-antipodality graph of $G$. If $\Upper_Q$ contains no full antipodal triples, then 
	\begin{itemize}\itemsep0em
		\item[--] any weak $Q$-coloring $f$ is a proper coloring of $H$ and any proper coloring $f$ of $H$ which satisfies conditions \ref{item:a}, \ref{item:b},\ref{item:c}, \ref{item:d}, is a weak-$Q$-coloring;
		\item[--] the restriction $f^\cross$ of a weak $Q$-coloring $f$ on $V(H^\cross)$ is the unique proper coloring of $H^\cross$ determined by conditions \ref{item:a}, \ref{item:b},\ref{item:c}, \ref{item:d};
		\item[--] any weak $Q$-coloring $f$ uniquely restricts to a proper coloring of $H^\cross$ and $f$ restricts to a proper 2-coloring of $H^\intra$.
	\end{itemize} 
\end{theorem}
We now prove the characterization.
\begin{theorem}\label{lemma:characterization_1_only_if}
	Let $Q$ be a clique separator of $G$. If $\Upper_Q$ contains no full antipodal triple and $G$ is weak $Q$-colorable, then $G$ is strong $Q$-colorable.
\end{theorem}
\begin{proof} Since $\Upper_Q$ contains no full antipodal triples, the $Q$-skeleton $\mathcal{D}$ of $G$ is a partition of $\Gamma_Q$. Let $f$ be a weak $Q$-coloring of $G$ and let $H$ be the $Q$-antipodality graph of $G$. By Theorem~\ref{thm:restricted_col}, $f$ is a proper coloring of $H$. Hence, to prove that $f$ is a strong $Q$-coloring it suffices to prove that $|f(\Lambda)|\leq 2$ for every neighboring triple $\Lambda$ contained in $\Gamma_Q$. 
	
	Let $A=\{\gamma\in \Gamma_Q \ | \ f(\gamma)=\ell+1\}$. Since $f$ is a proper coloring of $H$, $A$ is an independent set of $H$. Exactly as \eqref{eq:stable2} in Theorem~\ref{thm:canonical}, 
	\begin{center}
		if $\Lambda$ is a neighboring triple such that $\Lambda\cap A\neq \emptyset$, then $\exists i\in [\ell]$ such that $\Lambda\subseteq D_i$.
	\end{center}
	By the same reasons given in the proof of~\eqref{eq:stable2}, if the statement above were not true, then there would be $i\in [\ell]$, $\gamma\in \Lambda\cap D_i$ and $u\in\Upper_Q\setminus\{u_i\}$ , such that $f(\gamma)=\ell+1$, and $\gamma\leftrightarrow u$. However, the latter two conditions cannot hold simultaneously because of property~\ref{item:d} of the weak $Q$-coloring $f$. We therefore conclude that the statement is true.  
	
	Next we observe that for every $D\in \mathcal{D}$ the restriction of $f$ to $H[D]$ is a proper 2-coloring of $H[D]$ (still by Theorem~\ref{thm:restricted_col}). We therefore conclude that $|f(\Lambda)|\leq 2$ for every neighboring triple $\Lambda$ contained in $D$ for every $D\in \mathcal{D}$. In particular for all those neighboring triples that intersect $A$. Hence it remains only to prove that 
	\begin{equation}\label{eq:last_triple}
		\text{$|f(\Lambda)|\leq 2$ for every neighboring triple $\Lambda\subseteq \Gamma_Q\setminus A$}.
	\end{equation}
	Notice that $f(\Lambda)\subseteq [\ell]$ for any such triple. Observe now that by conditions~\ref{item:a}, \ref{item:b}, and \ref{item:c}, each $\gamma\in \Gamma_Q\setminus A$ has the same color as one of its (at most two) upper bounds. Denote by $\nu(\gamma)$ the unique element of $\Upper_Q$ such that $f(\gamma)=f(\nu(\gamma))$. Hence, if~\eqref{eq:last_triple} were not true, then $|f(\Lambda)|=3$ for some neighboring triple $\Lambda\subseteq \Gamma_Q\setminus A$ implying $|f(\nu(\Lambda))|=3$, and $\gamma\leq \nu(\gamma)$ for $\gamma\in \Lambda$ (we have set $\nu(\Lambda)=\{\nu(\gamma) \, |\ \gamma\in \Lambda\}$). Since $\Lambda$ is neighboring triple, so is the triple $\nu(\Lambda)$ (by a repeated application of Lemma~\ref{lemma:elenco_trivial_1}). Since  $\nu(\Lambda)\subseteq \Upper_Q$, $\nu(\Lambda)$ is a full antipodal triple contained in $\Upper_Q$. Therefore, the hypothesis that $\Upper_Q$ contains no such a triple is contradicted and the proof is completed.
\end{proof}     
We summarize the characterization of path graphs within chordal graphs in the following result. For readers' convenience we listed explicitly the defining conditions of weak $Q$-colorability.
\begin{theorem}\label{cor:main}
	A chordal graph $G$ is a path graph if and only if, for all clique separators $Q$ of $G$, $\Upper_Q$ contains no full antipodal triple and $G$ is weak $Q$-colorable, namely, there exists $f:\Gamma_Q\rightarrow [\ell+1]$ satisfying the following:
	\begin{enumerate}[label={\rm(\alph*)}]\itemsep0em
		\item for all $i\in[\ell]$, $f(u_i)=i$, 
		\item for all $i\in[\ell]$, for all $\gamma\in D_i$, $f(\gamma)\in\{i,\ell+1\}$,
		\item for all $i<j\in[\ell]$, for all $\gamma\in D_{i,j}$, $f(\gamma)\in\{i,j\}$,
		\item for all $i\in[\ell]$, for all $\gamma\in D_i$ if $\exists u\in\Upper_Q$ such that $\gamma\leftrightarrow u$, then $f(\gamma)=i$,
		\item for all $i<j\in[\ell]$, for all $\gamma\in D_{i,j}$ such that $\exists\gamma'\in D_{k}$, for $k\in\{i,j\}$, satisfying $\gamma\leftrightarrow\gamma'$, then $f(\gamma)=\{i,j\}\setminus\{k\}$,
		\item for all $D\in\mathcal{D}$, for all $\gamma,\gamma\in D$ such that $\gamma\leftrightarrow\gamma'$, $f(\gamma)\neq f(\gamma')$, 
	\end{enumerate}
	where $\ell=|\Upper_Q|$, and $\mathcal{D}$ is the $Q$-skeleton of $G$.
\end{theorem}
By Theorem~\ref{thm:always_exists_partial}, weak $Q$-colorings uniquely restrict to proper colorings of $H^\cross$. Such colorings are maps that satisfy the first (suitably restricted) five conditions of Definition~\ref{def:partial_and_weak} on $V(H^\cross)$. On the other hand, still by the theorem, such maps are precisely the proper colorings of $H^\cross$ which satisfy the first four conditions of Definition~\ref{def:partial_and_weak}. 
Denote by $\ccro$ the unique proper coloring of $H^\cross$ satisfying the first (suitably restricted) four conditions of Definition~\ref{def:partial_and_weak}. These conditions read as follows,
\begin{enumerate}[label={\rm(\roman*)}]\itemsep0em
	\item\label{com:partial_i} $\ccro(\gamma)=i$ for all $\gamma\in V(H^\cross)\cap D_i$ and for all $i\in [\ell]$,
	\item\label{com:partial_ii} $\ccro(\gamma)\in\{i,j\}$ for all $i,\,j\in [\ell]$, $i<j$ and for all $\gamma\in D_{i,j}\cap V(H^\cross)$, 
	\item\label{com:partial_iii} for all $i\in[\ell]$, for all $\gamma\in D_i$ if $\exists u\in\Upper_Q$ such that $\gamma\leftrightarrow u$, then $\ccro(\gamma)=i$,
\end{enumerate}
where \ref{com:partial_iii} is the same as condition \ref{item:d} because, as noticed after Remark~\ref{rem:tau}, conditions~\ref{item:d} and \ref{item:e} apply to $H^\cross$. Observe that there always exists a map $f: V(H^\cross)\rightarrow [\ell]$ which satisfies conditions \ref{com:partial_i}, \ref{com:partial_ii}, and \ref{com:partial_iii}. Moreover, such a map, has a unique restriction $\cro_Q$ on $V(H^\cross)\cap(\cup_iD_i)\rightarrow [\ell]$ defined by  $\cro_Q(\gamma)=i\Leftrightarrow \gamma\in V(H^\cross)\cap D_i$. Notice that $\cro_Q$ is a proper coloring of the subgraph of $H^\cross$ induced by $V(H^\cross)\cap(\cup_iD_i)$: if $\gamma\in D_i$ and $\gamma'\in D_j$ are such that $\gamma\leftrightarrow\gamma'$ for some $i,j$ with $i<j$, then $\gamma\leftrightarrow u_j$ and $\gamma'\leftrightarrow u_i$ both hold by Lemma~\ref{lemma:elenco_trivial_2}.~\ref{item:antipodal_square}; hence, $\cro_Q(\gamma)=i\not=\cro_Q(\gamma')=j$ because of condition~\ref{com:partial_iii}.

Hence, if coloring $\ccro$ exists, then it extends $\cro_Q$ on all of $V(H^\cross)$ under the condition that $\ccro$ is a proper coloring extension or, equivalently, that $\ccro$ satisfies condition \ref{item:e}. Since \ref{item:e} cannot be satisfied only if $\Gamma_Q$ contains three elements $\gamma\in D_{i,j}$, $\gamma'\in D_i$ and $\gamma''\in D_j$, $i<j$, such that $\gamma\leftrightarrow \gamma'$, and  $\gamma\leftrightarrow \gamma'$, it follows that the absence of such triples is a sufficient condition for $\ccro$ to extend $\cro_Q$.

The next result summarizes the remarks above. It can be seen as the converse of Theorem~\ref{thm:restricted_col}. As usual, $\{D_i\ | i\in [\ell]\}\cup \{D_{i,j} \ | i,\,j\in [\ell], i<j\}$ is the $Q$-skeleton of $G$.
\begin{theorem}\label{thm:always_exists_partial}
	Let $Q$ be a clique separator of $G$ and $H$ the $Q$-antipodality graph of $G$. 
	\begin{itemize}\itemsep0em
		\item If $\Upper_Q$ contains no full antipodal triple, then there exists a unique map $\cro_Q$ which satisfies \ref{com:partial_i}, \ref{com:partial_ii}, and \ref{com:partial_iii} on $V(H^\cross)\cap(\cup_iD_i)$. Such a map is a proper coloring of the subgraph induced by $V(H^\cross)\cap(\cup_iD_i)$ in $H^\cross$.  
		\item If, in addition, $\Gamma_Q\setminus \Upper_Q$ contains no triple of the form
		\begin{equation}\label{eq:bad_p3} 
			\text{$\{\gamma,\gamma',\gamma''\}$, $\gamma\in D_{i,j}$, $\gamma'\in D_i$, and $\gamma''\in D_j$ such that $\gamma\leftrightarrow \gamma'$, and  $\gamma\leftrightarrow \gamma''$, $i,\,j\in [\ell]$, $i<j$},
		\end{equation}
		then there is a unique extension $\ccro$ of $\cro_Q$ to all of $V(H^\cross)$ which satisfies condition~\ref{item:e}. Such a map is a proper coloring of $H^\cross$.
		\item If, in addition, $G$ is weak $Q$-colorable, then any weak $Q$-coloring $f$ extends $\ccro$ to a proper coloring of $H$. Moreover, $f$ restricts to a proper 2-coloring of $H^\intra$.  
	\end{itemize}
\end{theorem}
Although  the result above is seemingly tautological, it makes really apparent the effect of the stiff structure of antipodal triple free attachedness graphs in the recognition of path graphs. In the first place, maps $\cro_Q$ and $\ccro$ can be viewed 
as ``partial proper colorings'' of $H$, namely, proper colorings of subgraphs of $H$. Therefore, weak-colorings can be sought among the extensions of these partial coloring rather than constructed from scratch and it can be easily decided whether such partial colorings exist because, in case, they are uniquely determined by the structure of $H$. In the second place, modulo the map $\ccro$---uniquely determined by $H^\cross$---strong $Q$-colorings are nothing but proper 2-colorings of $H^\intra$ where some of the vertices are already colored. In this precise sense we have reduced strong $Q$-colorability, namely, a constrained vertex coloring problem, to a 2-colorability problem. This is the essence of our ``weakening''. As a consequence of Theorem~\ref{cor:main} and Theorem~\ref{thm:always_exists_partial} we have the following
\begin{corollary}\label{cor:above}
	a chordal graph $G$ is not a path graph if and only if there exists a clique separator $Q$ such that either one of the following applies
	\begin{itemize}
		\item $\Upper_Q$ contains a full antipodal triple,
		\item $\Gamma_Q\setminus \Upper_Q$ contains a triple of the form given in~\eqref{eq:bad_p3},
		\item the restriction of $\ccro$ to $V(H^\cross)\cap D$ for some member $D$ of the $Q$-skeleton of $G$, does not extend to  a proper 2-coloring of $H[D]$, where $H$ is the antipodality graph of $G$.
	\end{itemize} 
\end{corollary}

\subsection{Consequences of the characterization}\label{sec:consequences}
Theorem~\ref{cor:main} gives the so-called \emph{good characterization} for path graph membership within chordal graphs, namely, it proves that path graph membership of chordal graphs is in $\text{NP}\cap\text{CoNP}$ without resorting to existing polynomial-time algorithms.~To the best of our knowledge, this is the first such characterization. The ``only if'' part of Theorem~\ref{cor:main} is a short proof of path graph membership: since a chordal graph has linearly many clique separators, for each such separator $Q$, it suffices to exhibit the triplet $(\mathcal{D}^Q,H^Q,f)$ where 
$\mathcal{D}^Q$ is the $Q$-skeleton of $Q$, $H^Q$ is the $Q$-antipodality graph of $G$, and $f$ is a weak $Q$-coloring. Notice that both the $Q$-skeleton and the $Q$-antipodality graph are part of the proof. Hence, one has to prove that the former is a partition of $\Gamma_Q$ while the latter is the graph of the relation $\leftrightarrow$. However, both tasks can be accomplished in polynomial-time (see below). 

On the other hand, Corollary~\ref{cor:above} immediately provides a short refutation of path graph membership, namely, a short proof that a graph is not a path graph. The main consequence of Corollary~\ref{cor:above} is the characterization of path graphs by two exhaustive lists of obstructions to path graph membership in the form of minimal forbidden induced/partial 2-edge colored subgraphs in each of the \emph{attachedness graphs} of the input graph. Section~\ref{section:forbidden_subgraphs} is entirely devoted to such a characterization while the rest of the present section is devoted to further algorithmic consequences of the characterization. 
\mybreak
Theorem~\ref{cor:main} directly implies a polynomial-time algorithm for path graph membership (actually, as shown in \cite{mew,schaffer}, also for path graph realization with a little extra effort) which can be described as follows. For each clique separator $Q$, 
\begin{itemize}\itemsep0em
	\item compute the $Q$-skeleton and prove or refuse that the $Q$-skeleton is a partition of $\Gamma_Q$; this can be done in polynomial-time using the fact that $\leq$ is a partial order; moreover, one checks that the $Q$-skeleton is a partition of $\Gamma_Q$ either by computing the union or, roughly, by checking the triples in $\Upper_Q$ (these are essentially equivalent tasks because of Lemma~\ref{remark:D_i_is_a_partition});
	\item build the $Q$-antipodality graph of $G$; this task can be accomplished in polynomial time because we can answer the question $\gamma\leftrightarrow \gamma'$? roughly by comparing  the (polynomially many) relevant cliques of $\gamma$ and $\gamma'$ (each $\gamma$ is a a chordal graph);
	\item exhibit a weak $Q$-coloring or declare that none exists; the latter task can be accomplished in polynomial-time by extending the unique partial coloring $\cro_Q$ defined in Theorem~\ref{thm:always_exists_partial} to a proper coloring of the $Q$-antipodality graph of $G$; checking if such an extension exists, requires two simple tests implied by conditions \ref{item:e} and \ref{item:f}, respectively: 
	\begin{itemize}\itemsep0em
		\item[--] condition~\ref{item:e} cannot be satisfied only if only $\Gamma_Q$ contains a triple of the form described in~\eqref{eq:bad_p3}; this follows by Theorem~\ref{thm:always_exists_partial}; if \ref{item:e} is satisfied, then $\ccro$ is computed.
		\item[--] condition~\ref{item:f} cannot be satisfied only if there is some member $D$ of the $Q$-skeleton of $G$ such that the graph induced by $D$ in the $Q$-antipodality graph of $G$ cannot be 2-colored given the set of elements already colored by $\ccro$. 
	\end{itemize}
\end{itemize}
This algorithm admits a recursive implementation on the clique separators of the input graphs---in the same spirit as Theorem~\ref{thm:mw} and Sch\"{a}ffer's backtracking algorithm~\cite{schaffer}---. 

In each recursion step, one has to perform the three tasks described above. The overall complexity is determined by the complexity of these tasks. Among them, the construction of the $Q$-antipodality graph is the most expensive. A brute force construction (as the one described above) leads to an overall time- complexity which is worse than the time-complexity of the existing algorithms~\cite{chaplick,schaffer}. 

However, it is shown in~\cite{balzotti}, that by exploiting the very stiff structure of antipodality relation on the $D_i$'s and $D_{i,j}$'s described Lemma~\ref{lemma:elenco_trivial_2}---provided that  $\Upper_Q$ has no full antipodal triple---one can compute the antipodality graph quickly. Refer to~\cite{balzotti} for further details. 

The algorithm in \cite{balzotti} looks simpler and more intuitive than Sch\"{a}ffer's backtracking algorithm~\cite{schaffer} and requires no complex data structure in contrast to Chaplick's algorithm~\cite{chaplick}, yet it achieves the same time-complexity as that of the latter two algorithms.

\begin{remark}
The two defining conditions of strong $Q$-colorability are in trade-off. Weak $Q$-colorings solve this trade-off by minimizing the number of colors used within each member and between each pair of members of the $Q$-skeleton while keeping fixed the total number of colors. 
\end{remark}

The $Q$-skeleton of $G$ constitutes a powerful device also for the recognition of directed path graphs. It allows a quick construction of the $Q$-antipodality graph of directed path graphs as well. This leads to the algorithm in~\cite{balzotti} to recognize directed path graphs by relying on Monma and Wei~\cite{mew}' characterization of directed path graphs. This characterization is stated below in our terminology.
\begin{theorem}[Monma and Wei~\cite{mew}]
	A chordal graph $G$ is a directed path graph if and only if either $G$ is an atom or, for each clique separator $Q$ each graph $\gamma\in \Gamma_Q$ is a directed path graph and the $Q$-antipodality graph is 2-colorable.
\end{theorem}
Here an \emph{atom} is a directed path graph if and only if it is chordal (recall that an atom is a graph with no clique separator). The algorithm in~\cite{balzotti} is the first algorithm that uses Monma and Wei's characterization and that recognizes directed path graphs without using the results in~\cite{chaplick-gutierrez}. Still refer to~\cite{balzotti} for more details.

\section{Forbidden subgraphs in attachedness graphs}\label{section:forbidden_subgraphs}
We now reap the graph theoretic crops of Theorem~\ref{cor:main} by listing all the obstructions to strong $Q$-colorability in the form of subgraphs of the $Q$-attachedness graphs of a chordal graph $G$. Recall that the $Q$-attachedness graph of $G$ is the graph $(\Gamma_Q,\Join)$ with reflexive pairs neglected ---whose edges are therefore pairs $\gamma\gamma'$, $\gamma,\gamma'\in \Gamma_Q$ such that $\gamma\Join\gamma'$---. Also recall that the $Q$-antipodality and the $Q$-dominance graph of $G$ factor $(\Gamma_Q,\Join)$. Such a factorization yields a 2-edge coloring of $(\Gamma_Q,\Join)$ which models the interactions between $\leftrightarrow$ and $\leq$. 

We first describe the uncolored version of our obstructions to path graphs membership.

\begin{definition}\,\,\\
-- For an integer $m$ such that $m\geq 3$, the $m$-\emph{wheel} is the graph on $[m+1]$ where the vertices in $[m]$ induce a cycle and vertex $m+1$ is adjacent to all the other vertices (see Figure~\ref{fig:wheelsandfans}.a).\\
-- For an integer $m$ such that $m\geq 4$, the \emph{$m$-fan} is the graph on $[m]$ such that $[m-1]$ induces a path having end-vertices $1$ and $m-1$ and vertex $m$ is adjacent to all the other vertices (see Figure~\ref{fig:wheelsandfans}.b).\\
-- The \emph{$m$-chorded fan} is the graph obtained from the $m$-fan by adding an edge between vertices $1$ and $m-1$. Notice that the $m$-chorded fan is isomorphic to the $m-1$-wheel (see Figure~\ref{fig:wheelsandfans}.c).\\
-- For an integer $m$ such that $m\geq 4$, the \emph{$m$-double fan} is the graph on $[m]$ such that $[m]$ induces a cycle and vertices $m-1$ and $m$ are adjacent to all other vertices (see Figure~\ref{fig:wheelsandfans}.d).
\end{definition}

\begin{figure}[h]
\captionsetup[subfigure]{justification=centering}
\centering
	\begin{subfigure}{2.6cm}
\begin{overpic}[width=2.6cm,percent]{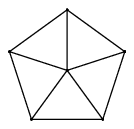}
\end{overpic}
\caption{}\label{fig:non-crossing_a}
\end{subfigure}
\qquad
	\begin{subfigure}{2.6cm}
\begin{overpic}[width=2.6cm,percent]{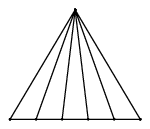}
\end{overpic}
\caption{}\label{fig:non-crossing_b}
\end{subfigure}
\qquad
	\begin{subfigure}{2.6cm}
\begin{overpic}[width=2.6cm,percent]{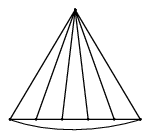}
\end{overpic}
\caption{}\label{fig:non-crossing_b_bis}
\end{subfigure}	
\qquad
	\begin{subfigure}{2.6cm}
\begin{overpic}[width=2.6cm,percent]{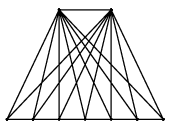}
\end{overpic}
\caption{}\label{fig:non-crossing_c}
\end{subfigure}	
  \caption{(a) $5$-wheel; (b) $7$-fan; (c) $7$-chorded-fan; (d) $9$-double fan.}
\label{fig:wheelsandfans}
\end{figure}

\noindent
Figure~\ref{fig:ostruzioni_indotte} lists certain special 2-edge-colored graphs, obtained as 2-edge-colored versions of the aforesaid graphs, needed in the characterization of path graphs (Theorem~\ref{cor:all}). The two colors are represented by dotted or solid lines, respectively.

\begin{figure}[h]
\centering
\begin{overpic}[width=10cm]{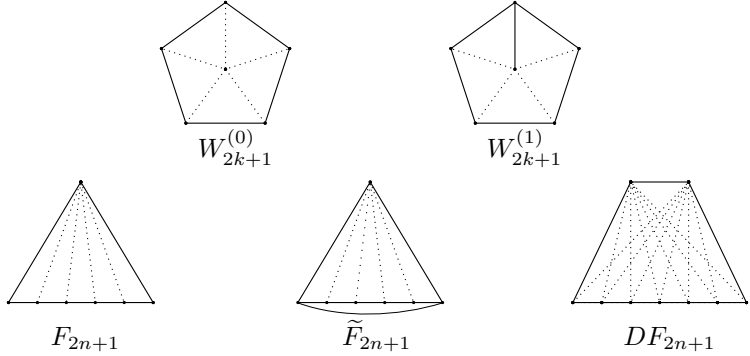} 
\put(26,26){$W^{(0)}_{2k+1}$}    
\put(64,26){$W^{(1)}_{2k+1}$}  
\put(6.7,1){$F_{2n+1}$}  
\put(45,1){$\widetilde{F}_{2n+1}$}  
\put(82,1){$DF_{2n+1}$}         
\end{overpic}
\caption{2-edge-colored graphs occurring in Theorem~\ref{cor:all}, $k\geq1$ and $n\geq2$.}
\label{fig:ostruzioni_indotte}
\end{figure}

It is convenient to settle a specific notation and terminology to present the results. An \emph{isomorphism of edge-colored graphs} is a graph isomorphism which preserves edge colors. All of the 2-edge-colored graphs in Figure~\ref{fig:ostruzioni_indotte} are pairwise non isomorphic as edge-colored graphs. We denote by $\mathcal{F}$ the collection they form ---$\mathcal{F}$ stands for “forbidden”---. Hence
$$\mathcal{F}=\left\{W_{2k+1}^{(0)},\,W_{2k+1}^{(1)},\, F_{2n+1}, \, \widetilde{F}_{2n+1},\, DF_{2n+1} \ |\ k\geq 1,\,n\geq 2\right\}.$$
Also let 
$$\mathcal{F}_0=\left\{W_{2k+1}^{(0)},W_{2k+1}^{(1)},F_{2n+1}\right\}.$$
Triangles of attachedness graphs play a special role. A triangle which is induced by a neighboring triple the $Q$-attachedness graph of $G$ is called a \emph{full triangle}, otherwise it is called \emph{empty}. A triangle all whose edges are antipodal is an \emph{antipodal triangle}. Not every triangle in $Q$-attachedness graph of $G$ is full, indeed an antipodal triangle might be empty (recall the discussion right after Lemma~\ref{lemma:elenco_trivial_1}). Unfortunately there is no way to establish whether an antipodal triangle is full or empty\footnote{Let $G$ be the graph $F_2$ in Figure~\ref{fig:lmp}; $G$ has only one separator, $Q$, say; let $H$ and $M$ be its $Q$-antipodality and $Q$-attachedness graphs; hence $M=H\cong K_3$ and the triangle spans a neighboring triple. However, if we denote by $z$ the universal vertex of $G$, then $G-z$  is separated by $Q\setminus z$. Again, let $Q'=Q\setminus z$ be the only clique separator; it holds that $M=H\cong K_3$ but the triangle does not span a neighboring triple.}. We know that full antipodal triples are obstructions to strong $Q$-colorability. Therefore, full antipodal triangles are obstructions to membership in the class of path graphs and they should be added to $\mathcal{F}$. However, since full antipodal triangles are not just edge colored triangles (because they have also the property of being full), we must treat such triangles separately in our statements. To overcome this (somehow unaesthetic and boring) ambiguity we use a standard trick.
\mybreak
For a graph $G$ let $G^+$ be the graph defined as follows. Let $V(G)=V=\{v_1,v_2,\ldots,v_n\}$ and $V^+$ be a copy of $V$, $V^+=\{v_1^+,v_2^+,\ldots,v_n^+\}$. Let
\begin{eqnarray}\label{eq:G+}
	G^+=\left(V\cup V^+,E(G)\cup\{v_iv_i^+\}_{i=1}^n\right).
\end{eqnarray}

\begin{lemma}\label{lemma:trick} Let $G$ be a graph. Then $G$ is a path graph if and only if $G^+$ is a path graph.
\end{lemma}
\begin{proof}
	Since $G$ is an induced subgraph of $G^+$, $G$ is a path graph if $G^+$ is such.
	Let $T$ be a clique path tree of $G$. For all $v\in V(G)$, let $\mathcal{K}_v$ the set of all cliques of $G$ containing $v$. By Theorem~\ref{thm:gavril}, $\mathcal{K}_v$ induces a path in $T$, let $\tilde{Q}_v\in \mathcal{K}_v$ be an end-vertex of this path. Thus it suffices to join $vv^+$ to $\tilde{Q}_v$ for all $v\in V(G)$ to yield a clique path tree for $G^+$.
\end{proof}
The reason for having introduced graph $G^+$ relies on the fact that, for every clique separator $Q$ of $G^+$, full antipodal triangles of $G$ display in $Q$-attachedness graph of $G^+$ as small wheels as shown next. 
\begin{lemma}\label{lemma:g+}
	Let $Q$ be a clique separator of $G^+$ and let $M$ be the $Q$-attachedness graph of $G^+$. Then $M$ has no full antipodal triangle and has no induced copy of $W_{2k+1}^{(0)}$if and only if $M$ has no induced copy of $W_{2k+1}^{(0)}$.
\end{lemma}
\begin{proof}
	One direction is trivial. For the other direction it suffices to prove that if $M$ has no induced copy of $W_{3}^{(0)}$, then $M$ has no full antipodal triangles. We prove the contrapositive: if $M$ has a full antipodal triangle, then $M$ has an induced copy of $W_{3}^{(0)}$. Observe first that $Q\cap \{v^+ \ |\ v\in V(G)\}=\emptyset$. For if not, then $Q$ is necessarily of the form $\{v,v^+\}$ for some $v\in V(G)$ (notice that in this case $v$ is a cut vertex); in this case however $M$ would contain no antipodal edges at all and thus no full antipodal triangles. Hence $v^+\not\in Q$ for each $v\in Q$. Notice that for each $v\in Q$ the graph $\gamma^+=(\{v,v^+\},\{vv^+\})$ is $\leq$-dominated by every other neighboring subgraph $\gamma$ of $v$. Let $\{\gamma,\gamma',\gamma''\}$ be the set of vertices of a full antipodal triangle in $M$. Hence, there is some $z\in V(G)$ such that $\gamma$, $\gamma'$ and $\gamma''$ are neighboring subgraphs of $z$. If  $\gamma_z$ is the subgraph of $G$ induced by $\{z,z^+\}\cup Q$, then $\{\gamma_z,\gamma,\gamma',\gamma''\}$ induces a copy of $W_3^{(0)}$ in $M$.
\end{proof}
In the following theorem we claim our characterization by forbidden subgraphs in the attachedness graphs. Note that graphs in $\mathcal{F}$ are induced obstructions, while graphs in $\mathcal{F}_0$ are not necessarily
induced. Moreover, statements~\ref{itfin:nof0+} and \ref{itfin:nof+} are equivalent to~\ref{itfin:nof} and \ref{itfin:nof}, respectively, by using $G^+$ in place of $G$ thanks to Lemma~\ref{lemma:trick} and Lemma~\ref{lemma:g+}.

\begin{theorem}\label{cor:all}
Let $G$ be a chordal graph. Then the following statements are equivalent:
	\begin{enumerate}[label=\subscript{\alph*}{\rm )}, ref=\subscript{\alph*}]\itemsep0em
		\item\label{itfin:G_path_graph} $G$ is a path graph,
		\item\label{itfin:nof0} for every clique separator $Q$ of $G$, the $Q$-attachedness graph of $G$ has no full antipodal triangle and has no subgraph isomorphic to any of the graphs in $\mathcal{F}_0$,
		\item\label{itfin:nof0+} for every clique separator $Q$ of $G$, the $Q$-attachedness graph of $G^+$ has no subgraph isomorphic to any of the graphs in $\mathcal{F}_0$,
		\item\label{itfin:nof} for every clique separator $Q$ of $G$, the $Q$-attachedness graph of $G$ has no full antipodal triangle and has no induced subgraph isomorphic to any of the graphs in $\mathcal{F}$,
		\item\label{itfin:nof+} for every clique separator $Q$ of $G$, the $Q$-attachedness graph of $G^+$ has no induced subgraph isomorphic to any of the graphs in $\mathcal{F}$.
	\end{enumerate}
\end{theorem}
The equivalences \ref{itfin:nof0}$\Leftrightarrow$\ref{itfin:nof0+} and \ref{itfin:nof}$\Leftrightarrow$\ref{itfin:nof+} in the theorem above follow straightforwardly by Lemma~\ref{lemma:trick} and Lemma~\ref{lemma:g+}. The remaining implication in Theorem~\ref{cor:all} (the core of the characterization), will be the content of the next section. We close this section instead with a brief comparison of our characterization with L\'{e}v\^{e}que, Maffray,~and~Preissmann's characterization~\cite{bfm}. 

Table~\ref{table:syn} gives a kind of dictionary between the two characterizations. The table reads as follows. For each row of the table, if a chordal graph $G$ contains an induced copy of one of the subgraphs in the leftmost column (according to L\'{e}v\^{e}que, Maffray,~and~Preissmann's characterization), then each of the graphs in the rightmost column occurs as an induced copy in the $Q$-attachedness graph of $G^+$ for some clique separator $Q$ (according to our characterization). 

Observe that for each graph $F$ in the leftmost column Table~\ref{table:syn}, there is no need to build the graph $F^+$ because, for each clique separator $Q$ of $F$, a full antipodal triangle in the $Q$-attachedness of $F$ corresponds to $W^{(0)}_3$ in the $Q$-attachedness graph of $F^+$. Let us give some clue on the content of the table. Obstructions $F_i$ for $i\in\{1,2,3,4,6,7,13,14,15\}$ have exactly one clique separator and thus there is one to one correspondence between L\'{e}v\^{e}que, Maffray,~and~Preissmann's obstructions and ours. Obstructions $F_j$ for $j\in\{8,9,11,16\}$ have exactly two clique separators but they generate the same obstruction in $\mathcal{F}$ by symmetry. The same happens for obstructions $F_5(n)$ and $F_{10}(n)$, where the number of clique separators grows with $n$ but all clique separators generate similar attachedness graphs that have the same obstruction. Obstruction $F_{12}(4k)$ deserves special attention because it has two clique separators that generate two different attachedness graphs. Furthermore, cases $k=2$ and  $k>2$ have to be distinguished.

From the table it is apparent a sort of coarsening of the obstructions. However we remark that the obstructions in our characterization are 2-edge colored subgraphs and that they have to be forbidden in each graph of the collection of the attachedness graph of $G^+$, while in L\'{e}v\^{e}que, Maffray,~and~Preissmann's characterization the obstructions are forbidden in the input graph itself.

\begin{figure}[h]
\centering
\begin{overpic}[width=15cm]{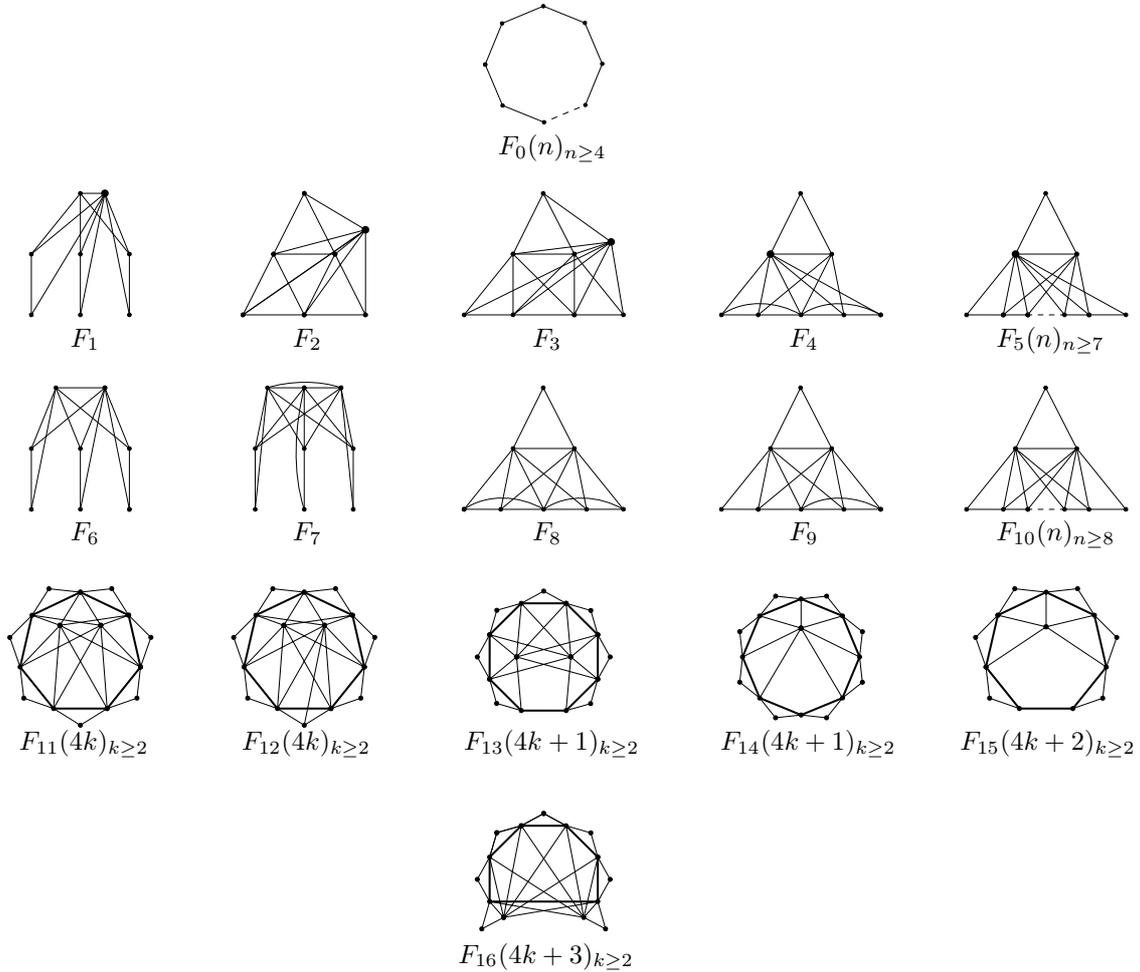}

\put(43.5,74.5){$F_0(n)_{n\geq4}$}

\put(6,57.5){$F_1$}
\put(25.5,57.5){$F_2$}
\put(46.5,57.5){$F_3$}
\put(69,57.5){$F_4$}
\put(87.2,57.5){$F_5(n)_{n\geq7}$}

\put(6,40.3){$F_6$}
\put(25.5,40.3){$F_7$}
\put(46.5,40.3){$F_8$}
\put(69,40.3){$F_9$}
\put(87.2,40.3){$F_{10}(n)_{n\geq8}$}

\put(1.5,21.8){$F_{11}(4k)_{k\geq2}$}
\put(21,21.8){$F_{12}(4k)_{k\geq2}$}
\put(40.5,21.8){$F_{13}(4k+1)_{k\geq2}$}
\put(63,21.8){$F_{14}(4k+1)_{k\geq2}$}
\put(84,21.8){$F_{15}(4k+2)_{k\geq2}$}

\put(40,3){$F_{16}(4k+3)_{k\geq2}$}
\end{overpic}
\caption{L\'{e}v\^{e}que, Maffray and Preissmann's exhaustive list of minimal non path graphs~\cite{bfm} (bold edges form a clique).}
\label{fig:lmp}
\end{figure}

\begin{table}[h]
	\begin{center}
		{\tablinesep=0.01ex\tabcolsep=2pt
			\begin{tabular}{||c c||} 
				\hline
				Family & Obstruction  \\
				\hline\hline
				$F_1, F_2,\ldots, F_9, F_{10}$ & $W_3^{(0)}$ \\ 
				\hline
				$F_{11}(4k)_{k\geq2}$ & $W_{2k-1}^{(0)}$\\
				\hline
				$F_{12}(4k)_{k\geq2}$ & $W_3^{(0)}$, $W_{3}^{(1)}$,  (for $k=2$), $F_{2k-1}$, $W^{(1)}_{2k-1}$ (for $k>2$)\\
				\hline
				$F_{13}(4k+1)_{k\geq2}$ &  $DF_{2k-1}$\\
				\hline
				$F_{14}(4k+1)_{k\geq2}$ &  $\widetilde{F}_{2k+1}$\\
				\hline
				$F_{15}(4k+2)_{k\geq2}$ &  $F_{2k+1}$\\
				\hline
				$F_{16}(4k+3)_{k\geq2}$& $F_{2k+1}$\\
				\hline
				\hline
			\end{tabular}
		}
		\caption{A dictionary between L\'{e}v\^{e}que, Maffray and Preissmann's characterization and Statement~\ref{itfin:nof} in Theorem~\ref{cor:all}. Note that $F_0$ is the obstruction to chordality.}
		\label{table:syn}
	\end{center}
\end{table}

\subsection{Proof of Theorem~\ref{cor:all}}
We now prove the core of Theorem~\ref{cor:all} according to the schema~\ref{itfin:G_path_graph}$\xLeftrightarrow{\text{Lemma~\ref{lemma:nofan},~\ref{lemma:main}}\,\,}$\ref{itfin:nof0}$\xLeftrightarrow{\text{Lemma~\ref{lemma:ind=sub}}\,\,}$\ref{itfin:nof}; we remember that the equivalences \ref{itfin:nof0}$\Leftrightarrow$\ref{itfin:nof0+} and \ref{itfin:nof}$\Leftrightarrow$\ref{itfin:nof+} are implied by Lemma~\ref{lemma:trick} and Lemma~\ref{lemma:g+}. In what follows $G$ is a chordal graph which is not an atom.

\begin{lemma}\label{lemma:nofan}
If $G$ is a path graph, then, for each clique separator $Q$, the $Q$-attachedness graph of $G$ has neither full antipodal triangles nor copies of any of the graphs in $\mathcal{F}_0$ as subgraphs.
\end{lemma}
\begin{proof}
Let $Q$ be a clique separator. Let us denote by $M$ the $Q$-attachedness graph of $G$. Since $G$ is a path graph, then $\Gamma_Q$ contains no full antipodal triangle by Lemma~\ref{lemma:nofull}. Suppose by contradiction that $M$ contains, as a subgraph, a copy $S$ of $F_{2n+1}$ or $W_{2k+1}^{(0)}$ or $W_{2k+1}^{(1)}$. In all cases, $S$ contains a subgraph $F_0$ on $\{\theta_0,\theta_1\dots,\theta_{2t}\}$, $t$ being a positive integer, fulfilling the following conditions
	\begin{itemize}\itemsep0em
		\item[--] $\theta_i\theta_{i+1}$ in an antipodal edge of $M$, namely $\theta_i\leftrightarrow\theta_{i+1}$, for $i=1,\ldots,2t-1$;
		\item[--] $\theta_0\theta_i$ is a dominance edge of $M$, namely, either $\theta_i\leq\theta_0$ or $\theta_0\leq\theta_i$, for all $i=1,\ldots,2t$.
	\end{itemize} 
We claim that:
	\cadre\label{lemma:scorciatoia}If $f$ is any strong $Q$-coloring of $G$, then $f(\theta_1)\neq f(\theta_{2t})$ and $f(\theta_0)\in\{f(\theta_1),f(\theta_{2t})\}$.
	\endcadre
	\begin{claimproof}{\eqref{lemma:scorciatoia}}
By Lemma~\ref{lemma:elenco_trivial_1} all triangles $\{\theta_0,\theta_i,\theta_{i+1}\}$ are full, for $i=1,\ldots,2t-1$. Hence ,being $f$ a strong $Q$-coloring, its second defining condition~\ref{com:mw_ii} implies that $|f(\{\theta_0,\theta_i,\theta_{i+1}\})|=2$, for $i=1,\ldots,2t-1$. Thus if $f(\theta_0)=f(\theta_1)$, then $f(\theta_2)\neq f(\theta_0)$, $f(\theta_3)=f(\theta_0),\dots, f(\theta_{2t})\neq f(\theta_0)$. Instead, if $f(\theta_0)\neq f(\theta_1)$, then $f(\theta_2)= f(\theta_0)$, $f(\theta_3)\neq f(\theta_0),\ldots, f(\theta_{2t})=f(\theta_0)$. In both cases, the thesis follows.
	\end{claimproof}
	We now use Claim~\eqref{lemma:scorciatoia} to prove a contradiction to the strong $Q$-colorability of $G$. Suppose first that $S\cong F_{2n+1}$, for some $n$, then let $V(S)=\{\eta,\gamma_1,\dots,\gamma_{2n}\}$ where $\eta$ is the maximum degree vertex of $S$. Let $F'$ be the subgraph induced by $V(S)=\{\eta,\gamma_2,\dots,\gamma_{2n-1}\}$. Hence $F'\cong F_0$. By Claim~\eqref{lemma:scorciatoia}, $\gamma_2$ and $\gamma_{2n-1}$  have opposite colors and $f(\gamma_{\eta})\in\{f(\gamma_2),f(\gamma_{2n-1})\}$. Moreover, the triangles induced by $\{\eta,\gamma_1,\gamma_2\}$ and $\{\eta,\gamma_{2n-1},\gamma_{2n}\}$ are both full by Lemma~\ref{lemma:elenco_trivial_1} and at least one of them cannot be 2-colored under $f$. 
	
	Suppose next that $S\cong W_{2k+1}^{(0)}$ or $S\cong W_{2k+1}^{(1)}$ for some $k$. Let $V(S)=\{\eta,\gamma_1,\dots,\gamma_{2k+1}\}$ where $\eta$ is still the maximum degree vertex of $S$ (if $S\cong W_{2k+1}^{(1)}$, then let $\gamma_1$ be the only vertex such that $\gamma_1\eta$ is an antipodal edge) and let $F''$ be the subgraph induced by $V(S)=\{\eta,\gamma_1,\dots,\gamma_{2k}\}$. Clearly, $F''\cong F_0$, as well. By Claim~\eqref{lemma:scorciatoia}, $\gamma_1$ and $\gamma_{2k}$  have opposite colors and $f(\eta)\in\{f(\gamma_1),f(\gamma_{2k})\}$. It holds that $f(\gamma_{2k+1})\not\in\{f(\gamma_{2k}),f(\gamma_1)\}$	
 because $\gamma_{2k+1}\leftrightarrow\gamma_{2k}$ and $\gamma_{2k+1}\leftrightarrow\gamma_1$. Moreover, the triangles induced by $\{\eta,\gamma_1,\gamma_{2k+1}\}$ and $\{\eta,\gamma_{2k},\gamma_{2k+1}\}$ are both full by Lemma~\ref{lemma:elenco_trivial_1} and at least one of them cannot be 2-colored under $f$. In any case a contradiction to the strong $Q$-colorability of $G$ is achieved. 
\end{proof}

\begin{lemma}\label{lemma:main}
If for each clique separator $Q$, the $Q$-attachedness graph of $G$ has neither full antipodal triangles nor copies of any of the graphs in $\mathcal{F}_0$ as subgraphs, then $G$ is a path graph.
\end{lemma}
\begin{proof}
By Corollary~\ref{cor:local_mew}, $G$ is a path graph if and only if $G$ is strong $Q$-colorable for each clique separator $Q$. We prove the contrapositive statement, namely, if $G$ is not strong $Q$-colorable for some clique separator $Q$, then the $Q$-attachedness graph $M$ of $G$ contains full antipodal triangles or some copy of a graph of $\mathcal{F}_0$ as subgraphs. Since each graph in $\mathcal{F}$ contains some graph of $\mathcal{F}_0$ as subgraph, we show the statement with $\mathcal{F}$ in place of $\mathcal{F}_0$. Let $H$ and $\mathcal{D}$ be the $Q$-antipodality graph, and the $Q$-skeleton of $G$, respectively. For $D\in \mathcal{D}$ denote by $H_D$ the subgraph of $H$ induced by $D$. 

By Corollary~\ref{cor:above}, if $G$ is not a path graph, then either $\Upper_Q$ contains a full antipodal triple, or $\Gamma_Q\setminus \Upper_Q$ contains a triple of the form described in~\eqref{eq:bad_p3} in Theorem~\ref{thm:always_exists_partial}, or there is $D\in \mathcal{D}$ such that the restriction of $\ccro$ to $D$ does not extend to a proper 2-coloring of $H_D$.

If $\Upper_Q$ contains a full antipodal triple, then this triple induces a full antipodal triangle in $M$. Hence we may assume that $\Gamma_Q$ contains no such triple. Suppose that $\Gamma_Q\setminus \Upper_Q$ contains a triple of the form $\{\gamma,\gamma',\gamma''\}$ with $\gamma\in D_{i,j}$, $\gamma'\in D_i$ and $\gamma''\in D_j$, for some $i\,j\in[\ell]$, $i<j$, such that $\gamma\leftrightarrow\gamma'$ and $\gamma\leftrightarrow\gamma''$. Thus either $\gamma'\leftrightarrow\gamma''$ or $\gamma'\not\leftrightarrow\gamma''$. If the first case applies, then $\{\gamma,\gamma',\gamma'',u_i\}$ induces a copy of $W^{(1)}_3$ in $M$ (refer to Lemma~\ref{lemma:elenco_trivial_1} and Lemma~\ref{lemma:elenco_trivial_2} to determine all colored edges in $M$). Else, the second case applies and $\{\gamma,\gamma',\gamma'',u_i,u_j\}$ induces a copy of $DF_5$.

Hence we may assume that neither $\Upper_Q$ contains full antipodal triples, nor $\Gamma_Q\setminus \Upper_Q$ contains a triple of the form described in~\eqref{eq:bad_p3} in Theorem~\ref{thm:always_exists_partial}. By Theorem~\ref{thm:always_exists_partial} there exists a unique proper coloring $\ccro$ of $H^\cross$. Let $\cro^D$ denote the restriction of $\ccro$ to $D$. If $\cro^D$ does not extend to proper 2-coloring of $H_D$, then only three cases can occur:
	\begin{itemize}\itemsep0em
		\item[--]  $H_D$ is non bipartite. In this case no 2-coloring of $H_D$ exists.
		\item[--] $H_D$ is bipartite but it contains a path $P$ with an even number of vertices whose endvertices have the same color under $\cro^D$. 
		\item[--] $H_D$ is bipartite but it contains a path $P$ with an odd number of vertices whose endvertices have different color under $\cro^D$. 
	\end{itemize}
\noindent
In the first case, $H_D$ contains an odd cycle $C$, on $2k+1$ vertices, say, as subgraph. Hence, for $u\in \Upper_Q\cap D$, the subgraph induced by $C\cup \{u\}$ in $H$ contains a copy of $W_{2k+1}^{(0)}$ as a subgraph.
	\mybreak
	In the second case let $\Theta=\{\theta_1,\ldots,\theta_{2k}\}$ be the set of vertices of $P$. Suppose first that $D=D_i$ for some $i\in [\ell]$. By definition of $\cro^D$ there are $\gamma,\gamma'\not\in D_i$ such that $\gamma\leftrightarrow\theta_1$ and $\gamma'\leftrightarrow\theta_{2k}$. It holds that $\gamma\leftrightarrow u_i$ and $\gamma'\leftrightarrow u_i$ by the transitivity of $\leq$ and the definition of $D_i$. Now, let $N$ be the subgraph induced by $\Theta\cup\{\gamma,\gamma',u_i\}$. If $\gamma=\gamma'$ then $N$ contains $W^{(1)}_{2k+1}$ as subgraph. If $\gamma\neq\gamma'$, then $N$ contains either $F_{2n+1}$ or $\widetilde{F}_{2n+1}$ according to whether $\gamma\leftrightarrow\gamma'$ or not. 
	If $D=D_{i,j}$, then we obtain the same results by a similar reasoning.
	\mybreak
	The third case can apply only to $D=D_{i,j}$ for some $i,\,j\in [\ell]$, because all the elements of $V(H^\cross)\cap D_i$ have the same color $i$ under $\cro^D$. Let $\Theta=\{\theta_1,\ldots,\theta_{2k+1}\}$ be the set of vertices of $P$. By the definition of $\cro^D$ there are $\gamma\in D_i$ and $\gamma'\in D_j$ such that $\gamma\leftrightarrow\theta_1$ and $\gamma'\leftrightarrow\theta_{2k+1}$. Then $\Theta\cup\{\gamma,\gamma',u_i,u_j\}$ induces a subgraph in $H_{i,j}$ that contains $DF_{2n+1}$ as subgraph.
\end{proof}

\begin{lemma}\label{lemma:ind=sub}
Let $Q$ be a clique separator of $G$. If the $Q$-attachedness graph of $G$ has no full antipodal triangle, then it has a copy of a graph in $\mathcal{F}_0$ as a subgraph if and only if it has a copy of a graph in $\mathcal{F}$ as an induced subgraph.
\end{lemma}
\begin{proof}
Since any graph in $\mathcal{F}_0$ is contained as subgraph in one of the graph in $\mathcal{F}$ one direction is trivial. Let us prove the other direction. Let $H$ and $M$ be the $Q$-antipodality  and $Q$-attachedness graph of $G$. We have to prove that if $M$ contains some copy of a graph of $\mathcal{F}_0$, then $M$ contains an induced copy of some graph of $\mathcal{F}$. Let $S$ be a graph of $\mathcal{F}_0$. For a cycle $C$ of $S$ it is convenient to distinguish between chords that are edges of the antipodality graphs, which we call $a$-chords, from those that are edges of the dominance graph, which we call $d$-chords. 

Let now $C$ be a \emph{antipodal} odd cycle of $S$ on $2k+1$ vertices for some integer $k\geq2$, i.e., the vertex set of $C$ is $\{\gamma_0,\ldots,\gamma_{2k}\}$ and the edges are $\{\gamma_0\gamma_1,\ldots,\gamma_{2k-1}\gamma_{2k},\gamma_0\gamma_{2k}\}$, where all the edges of $C$ are antipodal edges. Suppose that $C$ has either no $a$-chord, namely $C$ is induced in $H$, or $C$ has precisely the $a$-chord $\gamma_1\gamma_{2k}$. We will show that every graph in $\mathcal{F}_0$ contains such a cycle with possible $d$-chords with an end in $\gamma_0$. The following fact about such a $C$ is crucial to prove the lemma and it implies that if $C$ has at least one $d$-chord with an end in $\gamma_0$, then $C$ induces in $M$ a copy of $F_{2k+1}$, $DF_{2k+1}$ or $\widetilde{F}_{2k+1}$.
\cadre\label{eq:chords}
If $\gamma_0\gamma_j$ is a $d$-chord of $C$ with, say, $\gamma_j\leq \gamma_0$, $j\not\in\{1,2k\}$, then $C$ has $d$-chords $\gamma_0\gamma_l$ with $\gamma_j\leq\gamma_0$, for all $j\not\in\{1,2k\}$. Moreover,
\begin{itemize}\itemsep0em
\item if $C$ is induced in $H$ and $C$ has some other $d$-chord, then $C$ possesses either all $d$-chords $\gamma_1\gamma_j$ with $\gamma_j\leq \gamma_1$, $j\not \in \{0,2\}$, or, symmetrically, all the $d$-chords $\gamma_{2k}\gamma_j$, with $\gamma_j\leq \gamma_{2k}$, $j\not\in\{0,2k-1\}$,
\item if $\gamma_1\gamma_q$ is an $a$-chord of $C$, then $C$ has no other $d$-chords.  
\end{itemize}
\endcadre
\begin{claimproof}{\eqref{eq:chords}} In the first place, observe that $\gamma_{j-1}\leftrightarrow \gamma_j$ and $\gamma_{j+1}\leftrightarrow \gamma_j$ trivially imply $\gamma_{j-1}\Join\gamma_j$ and $\gamma_{j+1}\Join \gamma_j$ hence, by Lemma~\ref{lemma:elenco_trivial_1}, it holds that  $\gamma_0\Join\gamma_{j-1}$ and $\gamma_0\Join\gamma_{j+1}$. Thus $\gamma_0\gamma_{j-1}$ and $\gamma_0\gamma_{j+1}$ are $d$-chords of $C$, because the unique possible $a$-chord is $\gamma_1\gamma_{2k}$. Necessarily $\gamma_{j-1}\leq \gamma_0$ for, if not, then $\gamma_j\leq \gamma_0\leq \gamma_{j-1}$ would imply $\gamma_j\leq \gamma_{j-1}$ contradicting that  $\gamma_{j-1}\leftrightarrow\gamma_j$. By the same reasons, $\gamma_{j+1}\leq \gamma_0$. A repeated application of this argument to $j-1$ and $j+1$ in place of $j$ proves the first part of the claim (see Figure~\ref{fig:eptagons}(\subref{fig:eptagons_a})). 

The first part of the claim is clearly invariant under automorphisms of $C$. Consequently, we deduce that if $C$ has another $d$-chord  $\gamma_h\gamma_\ell$ with $\gamma_\ell\leq \gamma_h$ and $h\not\in\{1,2k\}$, then $C$ has also $d$-chords $\gamma_h\gamma_1$ and $\gamma_h\gamma_{2k}$. But this is impossible because it would imply $\gamma_{2k}\leq \gamma_h\leq \gamma_0$ while we know that $\gamma_0\leftrightarrow\gamma_{2k}$. Hence all the other possible $d$-chords of $C$ have one end in $\{\gamma_1,\gamma_{2k}\}$. On the other hand $C$ cannot possess $d$-chords $\gamma_1\gamma_h$ and $\gamma_{2k}\gamma_\ell$ for some $h,\,\ell\in [2k]$  because, by the first part of the claim, it would possess the $d$-chord $\gamma_1\gamma_{2k}$ and this would imply $\gamma_{2k}\leq \gamma_1$ and $\gamma_1\leq \gamma_{2k}$ and consequently the contradiction $\gamma_1=\gamma_{2k}$ (see Figure~\ref{fig:eptagons}(\subref{fig:eptagons_b}) and Figure~\ref{fig:eptagons}(\subref{fig:eptagons_c})). 
	
It remains to prove that if $\gamma_1\gamma_{2k}$ is an $a$-chord of $C$, then $C$ has no other $d$-chords with one end in $\{\gamma_1,\gamma_{2k}\}$ (hence no other $d$-chords at all, as Figure~\ref{fig:eptagons}(\subref{fig:eptagons_d})) shows). Suppose that $C$ has a $d$-chord with one end in $\{\gamma_1,\gamma_{2k}\}$, $\gamma_1$ say. Then $C$ has the $d$-chord $\gamma_1\gamma_{2k-1}$ by above. Since $\gamma_{2k-1}\leq\gamma_0$, $\gamma_{2k-1}\leq \gamma_1$ and $\gamma_{2k-1}\leftrightarrow \gamma_{2k}$, by Lemma~\ref{lemma:elenco_trivial_1} it follows that  $\{\gamma_0,\gamma_1,\gamma_{2k}\}$ induces a full antipodal triangle in $M$, contradicting that $M$ has no such triangles.   
\end{claimproof}
\begin{figure}[h]
\captionsetup[subfigure]{justification=centering}
\centering
\quad
	\begin{subfigure}{2.5cm}
\begin{overpic}[width=2.5cm,percent]{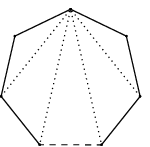}
\put(45,102){$\gamma_0$}
\put(-11,79){$\gamma_{2k}$}
\put(88,79){$\gamma_1$}

\put(100,35){$\gamma_2$}
\put(77,0){$\gamma_3$}
\put(-34,35){$\gamma_{2k-1}$}
\put(-9,0){$\gamma_{2k-2}$}
\end{overpic}
\caption{}\label{fig:eptagons_a}
\end{subfigure}
\qquad\qquad
	\begin{subfigure}{2.5cm}
\begin{overpic}[width=2.5cm,percent]{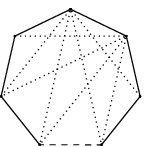}
\put(45,102){$\gamma_0$}
\put(-11,79){$\gamma_{2k}$}
\put(88,79){$\gamma_1$}

\put(100,35){$\gamma_2$}
\put(77,0){$\gamma_3$}
\put(-34,35){$\gamma_{2k-1}$}
\put(-9,0){$\gamma_{2k-2}$}
\end{overpic}
\caption{}\label{fig:eptagons_b}
\end{subfigure}
\qquad\qquad
	\begin{subfigure}{2.5cm}
\begin{overpic}[width=2.5cm,percent]{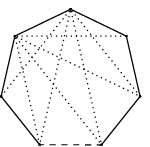}
\put(45,102){$\gamma_0$}
\put(-11,79){$\gamma_{2k}$}
\put(88,79){$\gamma_1$}

\put(100,35){$\gamma_2$}
\put(77,0){$\gamma_3$}
\put(-34,35){$\gamma_{2k-1}$}
\put(-9,0){$\gamma_{2k-2}$}
\end{overpic}
\caption{}\label{fig:eptagons_c}
\end{subfigure}
\qquad\qquad
	\begin{subfigure}{2.5cm}
\begin{overpic}[width=2.5cm,percent]{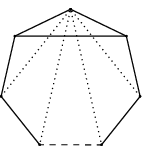}
\put(45,102){$\gamma_0$}
\put(-11,79){$\gamma_{2k}$}
\put(88,79){$\gamma_1$}

\put(100,35){$\gamma_2$}
\put(77,0){$\gamma_3$}
\put(-34,35){$\gamma_{2k-1}$}
\put(-9,0){$\gamma_{2k-2}$}
\end{overpic}
\caption{}\label{fig:eptagons_d}
\end{subfigure}
\caption{graphs in the proof of Claim~\ref{eq:chords}. Note that the graph in (\subref{fig:eptagons_a}) is isomorphic to $F_{2k+1}$, while the graphs in (\subref{fig:eptagons_b}) and (\subref{fig:eptagons_c}) are isomorphic to $DF_{2k+1}$, and the graph in (\subref{fig:eptagons_d}) is isomorphic to $\widetilde{F}_{2k+1}$.}\label{fig:eptagons}
\end{figure}
We can now complete the proof of the lemma. Let $S$ be a copy in $M$ of any of the three graphs in $\mathcal{F}_0$, and let $S$ have $n$ vertices $\gamma_0,\gamma_1,\ldots,\gamma_{n-1}$. Observe that $S$ possesses an odd cycle $R$ on at least $n-1$ vertices; more precisely, the wheels have an odd cycle on $n-1$ vertices and the fan on $n$ vertices. Let $\gamma_0$\ be the highest degree vertex in $S$ and let $H_R$ and $M_R$ be the graphs induced by $R$ in $H$ and $M$, respectively. Let $C$ be a cycle with minimum possible order among the odd cycles contained in $H_R$ whose order is at least 5. Hence either $C$ is an odd hole of $H$ or $C$ is an odd cycle of $H$ with exactly one $a$-chord which belongs to a triangle having the other two edges on $C$, otherwise the minimality is denied. Clearly, the dominance edges of $S$  induced by $V(C)$ are $d$-chords of $C$. Suppose first that $C$ has no extra $d$-chord other than those. In this case we are done because, $V(C)\cup \{\gamma_0\}$ (possibly $\theta\in V(C)$ when $S$ is a fan) induces either a wheel, or a fan, or a chorded fan. We may therefore assume that $C$ possesses some extra $d$-chord (a dominance edge of $M_R$ which is not in $S$). Possibly after relabeling, $C$ is of the form  described in Claim~\eqref{eq:chords} and $C$ possesses all the $d$-chords $\gamma_0\gamma_i$, $i\in [n-1]$ (by the claim). If $C$ possesses no other $d$-chords we are done, because $V(C)$ induces either a chorded fan or a fan according to whether or not $C$ possesses the unique $a$-chord $\gamma_1\gamma_t$. If $C$ possesses some other $d$-chord, still by Claim~\eqref{eq:chords}, then $C$ possesses either all $d$-chords $\gamma_1\gamma_j$ with $\gamma_j\leq \gamma_1$, $j\not \in \{0,2\}$, or all the $d$-chords $\gamma_t\gamma_j$, with $\gamma_j\leq \gamma_t$, $j\not\in\{0,t-1\}$. In this case $V(C)$ induces  a double fan in $M$. The proof is completed.
\end{proof}


\bibliography{biblio.bib}

\begin{thebibliography}{10}

\bibitem{balzotti}
{\sc L.~Balzotti}, {\em A {N}ew {A}lgorithm to {R}ecognize {P}ath {G}raphs and
  {D}irected {P}ath {G}raphs}, CoRR, abs/2012.08476 (2020).

\bibitem{cameron-hoang_1}
{\sc K.~Cameron, C.~T. Ho{\`{a}}ng, and B.~L{\'{e}}v{\^{e}}que}, {\em Asteroids
  in rooted and directed path graphs}, Electron. Notes Discret. Math., 32
  (2009), pp.~67--74.

\bibitem{cameron-hoang_2}
\leavevmode\vrule height 2pt depth -1.6pt width 23pt, {\em Characterizing
  {D}irected {P}ath {G}raphs by {F}orbidden {A}steroids}, J. Graph Theory, 68
  (2011), pp.~103--112.

\bibitem{chaplick}
{\sc S.~Chaplick}, {\em P{Q}{R}-{T}rees and {U}ndirected {P}ath {G}raphs},
  University of Toronto, 2008.

\bibitem{chaplick-gutierrez}
{\sc S.~Chaplick, M.~Gutierrez, B.~L{\'{e}}v{\^{e}}que, and S.~B. Tondato},
  {\em From {P}ath {G}raphs to {D}irected {P}ath {G}raphs}, in Graph Theoretic
  Concepts in Computer Science - 36th International Workshop, {WG} 2010,
  vol.~6410 of Lecture Notes in Computer Science, 2010, pp.~256--265.

\bibitem{dah}
{\sc E.~Dahlhaus and G.~Bailey}, {\em Recognition of {P}ath {G}raphs in
  {L}inear {T}ime}, in 5th Italian Conference on Theoretical Computer Science,
  World Scientific, 1996, pp.~201--210.

\bibitem{gavril1}
{\sc F.~Gavril}, {\em The {I}ntersection {G}raphs of {S}ubtrees in {T}rees are
  {E}xactly the {C}hordal {G}raphs}, J. Combinatorial Theory Ser. B, 16 (1974),
  pp.~47--56.

\bibitem{gavril_DV_algorithm}
\leavevmode\vrule height 2pt depth -1.6pt width 23pt, {\em A {R}ecognition
  {A}lgorithm for the {I}ntersection {G}raphs of {D}irected {P}aths in
  {D}irected {T}rees}, Discret. Math., 13 (1975), pp.~237--249.

\bibitem{gavril_UV_algorithm}
\leavevmode\vrule height 2pt depth -1.6pt width 23pt, {\em A {R}ecognition
  {A}lgorithm for the {I}ntersection {G}raphs of {P}aths in {T}rees}, Discret.
  Math., 23 (1978), pp.~211--227.

\bibitem{lekkerkerker-boland}
{\sc C.~Lekkekerker and J.~Boland}, {\em Representation of a finite graph by a
  set of intervals on the real line}, Fundamenta Mathematicae, 51 (1962),
  pp.~45--64.

\bibitem{bfm}
{\sc B.~L\'{e}v\^{e}que, F.~Maffray, and M.~Preissmann}, {\em Characterizing
  {P}ath {G}raphs by {F}orbidden {I}nduced {S}ubgraphs}, J. Graph Theory, 62
  (2009), pp.~369--384.

\bibitem{mew}
{\sc C.~Monma and V.~Wei}, {\em Intersection {G}raphs of {P}aths in a {T}ree},
  J. Combin. Theory Ser. B, 41 (1986), pp.~141--181.

\bibitem{panda}
{\sc B.~S. Panda}, {\em The forbidden subgraph characterization of directed
  vertex graphs}, Discret. Math., 196 (1999), pp.~239--256.

\bibitem{renz}
{\sc P.~Renz}, {\em Intersection {R}epresentations of {G}raphs by {A}rcs},
  Pacific J. Math., 34 (1970), pp.~501--510.

\bibitem{schaffer}
{\sc A.~Sch\"{a}ffer}, {\em A {F}aster {A}lgorithm to {R}ecognize {U}ndirected
  {P}ath {G}raphs}, Discrete Appl. Math., 43 (1993), pp.~261--295.

\bibitem{tarjan}
{\sc R.~Tarjan}, {\em Decomposition by {C}lique {S}eparators}, Discrete Math.,
  55 (1985), pp.~221--232.

\end{thebibliography}
\bibliographystyle{siam}

\end{document}